\documentclass[12pt]{article}
\usepackage{latexsym}
\usepackage[margin=1in]{geometry}
\usepackage{graphicx,enumerate} 
\usepackage{color} 
\usepackage{amsmath, amsthm, amssymb}
\usepackage{enumerate}
\usepackage{float}
\usepackage{url}
\usepackage{stackrel}
\usepackage{mathrsfs,dsfont}
\usepackage{caption}
\usepackage{subcaption}
\usepackage{wrapfig}
\usepackage{bbm}
\usepackage{bm}
\usepackage{tikz}
\usepackage[all,cmtip]{xy}
\usepackage{hyperref}

\newtheorem{theorem}{Theorem}[section]

\newtheorem{proposition}[theorem]{Proposition}

\newtheorem{conjecture}[theorem]{Conjecture}
\newtheorem{question}[theorem]{Question}
\newtheorem{procedure}[theorem]{Procedure}

\theoremstyle{definition}
\newtheorem{definition}[theorem]{Definition}

\newtheorem{example}[theorem]{Example}

\theoremstyle{remark}
\newtheorem{remark}[theorem]{Remark}

\newcommand{\St}{{S}}
\newcommand{\invtPoly}{\mathcal{P}}
\DeclareMathOperator{\im}{im}

\newcommand{\R}{\mathbb{R}}
\newcommand{\Rnn}{\mathbb{R}_{\geq 0}}

\def\rla{\rightleftarrows}
\def\been{\begin{enumerate}}
\def\enen{\end{enumerate}}

\newcommand{\Ktot}{K_{\mbox{tot}}}
\newcommand{\Ptot}{P_{\mbox{tot}}}
\newcommand{\Stot}{S_{\mbox{tot}}}

\definecolor{dgreen}{rgb}{.2,.6,.2}
\colorlet{darkgreen}{black!30!dgreen}

\definecolor{dblue}{rgb}{0.0,0.0,0.68}

\definecolor{myblue}{RGB}{120,10,10}

\begin{document}


\title{Emergence of oscillations in a \\ mixed-mechanism 
phosphorylation system}
\author{Carsten Conradi\footnote{HTW Berlin}, Maya Mincheva\footnote{Northern Illinois University}, and Anne Shiu\footnote{Texas A\&M University}}
\date{January 28, 2019}

\maketitle

\begin{abstract}
This work investigates the emergence of oscillations in one of the simplest cellular signaling networks exhibiting oscillations, 
namely, the dual-site phosphorylation and dephosphorylation network (futile cycle), in which 
the mechanism for phosphorylation
 is processive while the one for dephosphorylation is distributive (or vice-versa).
 The fact that this network yields oscillations was shown recently by Suwanmajo and Krishnan. 
Our results, which significantly extend their analyses, are as follows. 
First, in the three-dimensional space of total amounts, the border between systems 
with a stable versus unstable steady state is a surface defined by the vanishing of a single Hurwitz determinant.  Second, this surface consists generically of simple Hopf bifurcations. 
  Next, simulations suggest that 
when the steady state is unstable,   
  oscillations are the norm.
  Finally, the 
emergence of 
oscillations via a Hopf bifurcation is enabled by the catalytic and association constants
of the distributive part of the mechanism:
if these rate constants satisfy two
  inequalities, then the system generically admits a Hopf bifurcation.  
  Our proofs are enabled by the Routh-Hurwitz criterion, a
  Hopf-bifurcation criterion due to Yang, and a monomial
  parametrization of steady states. 
\end{abstract}

\section{Introduction}
Oscillations have been observed experimentally in signaling networks formed by phosphorylation
and dephosphorylation~\cite{yeast-mapk-oscillations,oscillations-mapk-cancer}, 
which suggests that these networks are involved in timekeeping and synchronization.  
Indeed, multisite phosphorylation is the main mechanism for establishing the 24-hour period in eukaryotic circadian clocks~\cite{ode,beat}.  
Our motivating question, therefore, is, How do oscillations arise in phosphorylation networks?  

We tackle this question for 
 the network that,
according to 
Suwanmajo and Krishnan,
 ``could be the simplest enzymatic modification scheme that can intrinsically exhibit oscillation''~\cite[\S 3.1]{SK}.
This network, in~\eqref{eq:mixed-network},
is the mixed-mechanism (partially processive, partially distributive) dual-site phosphorylation network (or 
 {\bf mixed-mechanism network} for short).  
Examples of networks that include both processive and distributive elements 
 include the ``processive model'' of 
  Aoki {\em et al.}~\cite[Table S2]{Aoki}
and a model of ERK regulation via enzymes MEK and MKP3~\cite[Fig.\ 2]{long-term}.

In the mixed-mechanism network, $S_i$ denotes a substrate with $i$ phosphate groups attached, and $K$ and $P$ are, respectively, a {\em kinase} and a {\em phosphatase} enzyme:
\begin{align}
  \label{eq:mixed-network}
  \begin{split}
S_{0}+K &  \underset{k_2}{\overset{k_1}\rla} S_{0}K \overset{k_3}\longrightarrow  S_1K \overset{k_4}\longrightarrow S_2+K  \\
S_{2}+P & \underset{k_6}{\overset{k_5}\rla} S_{2}P \overset{k_7}\longrightarrow  S_{1}+P \underset{k_9}{\overset{k_8}\rla} S_1P \overset{k_{10}}\longrightarrow S_0+P ~.
\end{split}
\end{align}
When the kinase {\em phosphorylates} -- that is, adds phosphate groups to -- a substrate 
in the mixed-mechanism network 
(via the reactions labeled by $k_1$ to $k_4$),
the kinase and substrate do {\em not} dissociate before both phosphate groups are added.  
Accordingly, the mechanism for phosphorylation is {\em processive}.  
In contrast, when the phosphatase {\em dephosphorylates} -- i.e., removes phosphate groups from -- a substrate (via reactions $k_5$ to $k_{10}$), this mechanism is {\em distributive}: the phosphatase and substrate dissociate each time a phosphate group is removed.  Accordingly, network~\eqref{eq:mixed-network} is said to have a mixed mechanism\footnote{Network~\eqref{eq:mixed-network} is symmetric to the mixed-mechanism network in which phosphorylation is distributive (instead of processive) and dephosphorylation is processive (instead of distributive), so our results apply equally well to that network (cf.\ \cite[networks 21--22]{SK}).}.

The dynamical systems arising from the mixed-mechanism network live in a 9-dimensional space, but, due to three conservation laws, are essentially 6-dimensional.  Specifically, the total amounts of kinase, phosphatase, and substrate -- denoted by $\Ktot$, $\Ptot$, and $\Stot$, respectively -- are conserved.  For each choice of three such total amounts and each choice of positive rate constants $k_i$, there is a unique positive steady state~\cite{SK}.
  One focus of our work is 
  determining when such a steady state undergoes 
  a Hopf bifurcation leading to oscillations (with any of
  the $k_i$'s or total amounts as bifurcation parameter).
 
\subsection{Summary of main results}

How do oscillations of the mixed-mechanism network emerge, and how robust are they?  
These questions are the motivation for our work.  Let us describe Suwanmajo and Krishnan's progress in this direction.  
They first found rate constants $k_i$ and total amounts, 
displayed in Table~\ref{tab:rates},
that yield oscillations~\cite[Supplementary Information]{SK}. 
\begin{table}[ht]
  \centering
  \begin{tabular}{|c|c|c|c|c|c|c|c|c|c|} \hline
    $k_{1}$ &  $k_{2}$ &  $k_{3}$ & $k_{4}$ &  $k_{5}$ &  $k_{6}$ & $k_{7}$ & $k_{8}$ &  $k_{9}$ &  $k_{10}$
     \\ \hline 
    1 & 1 & 1 & 1 & 100 & 1 & $0.9$ & 3 & 1 & 100 
    \\ \hline
  \end{tabular}
\quad
\begin{tabular}{|c|c|c|}
\hline
    $\Ktot$ 
	& $\Ptot$
	& $\Stot$\\
\hline
	 17.5 & 5 & 40
    \\ \hline
\end{tabular}
  \caption{\label{tab:rates}
    Rate constants (left) and total amounts (right), from~\cite[Supplementary Information]{SK}, which lead to oscillations 
   in the mixed-mechanism network~\eqref{eq:mixed-network}. }
\end{table}

Next, they examined whether oscillations persist as $\Ktot$ varies. 
What they found, summarized in Figure~\ref{fig:hopf}, is that oscillations persist when $\Ktot$ is in the (approximate) interval 
$(13.03, 29.23)$, 
 and oscillations arise as the unique steady state undergoes a Hopf bifurcation. 
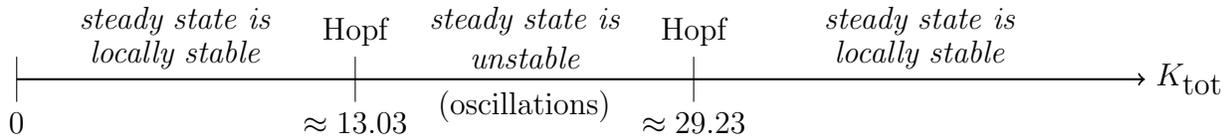
\begin{figure}[ht]
\begin{center}
	\begin{tikzpicture}[scale=3.0]
 \draw[thick,->] (0,0) -- (5,0);
\draw (0, -.1) -- (0,.1);
\draw (1.5, -.1) -- (1.5,.1);
\draw (3, -.1) -- (3,.1);
    	\node[below] at (0, -.1) {0};
    	\node[below] at (1.5,-.1) {$\approx 13.03$};
    	\node[below] at (3, -.1) {$\approx 29.23$};
	\node[right] at (5, 0) {$\Ktot$};
    	\node[above] at (0.7, 0.15) {{\em steady state is} }; 
   	\node[above] at (0.7, 0.0) {{\em locally stable}};
    	\node[above] at (2.25, 0.15) {{\em steady state is}}; 
    	\node[above] at (2.25, 0.0) {{\em unstable}};
    	\node[below] at (2.25, 0.0) {(oscillations)};
    	\node[above] at (4, 0.15) {{\em steady state is}}; 
    	\node[above] at (4, 0.0) {{\em locally stable}};
    	\node[above] at (1.5, 0.1) { Hopf};
    	\node[above] at (3, 0.1) { Hopf};
 	\end{tikzpicture}
\end{center}
\caption{Stability of the unique steady state of the mixed-mechanism network~\eqref{eq:mixed-network} as a function of $\Ktot$, as analyzed by Suwanmajo and Krishnan~\cite[Fig.\ 4]{SK}.  (The other total amounts, $\Ptot$ and $\Stot$, and the rate constants $k_i$ are those in~Table~\ref{tab:rates}.) 
Oscillations were found when $\Ktot$ is in the ``unstable'' interval~\cite{SK}.}
\label{fig:hopf}
\end{figure}

Subsequently, Conradi and Shiu~\cite{perspective} found that when $\Ptot$ also is allowed to vary,
oscillations exist for larger values of $\Ktot$ (e.g., $\Ktot=100$).  So, how exactly do oscillations depend on the three total amounts (or, equivalently, the initial conditions)?  
Concretely, our goal is to
expand Figure~\ref{fig:hopf} to encompass all possible perturbations to the initial conditions (i.e., the total amounts):
\begin{question} \label{q:osc}
Consider the mixed-mechanism network~\eqref{eq:mixed-network}, with 
$k_i$'s from Table~\ref{tab:rates}.
\begin{enumerate}
	\item For which values of $\left( \Ktot,~\Ptot,~\Stot \right) \in \mathbb{R}^3_{>0}$ is the unique steady state unstable?  
	\item Whenever (by perturbing parameters or total amounts) a steady state
          switches from being locally stable to unstable, 
            does this always 
            give rise to a Hopf bifurcation? 
\end{enumerate}
\end{question}

The direct method for solving Question~\ref{q:osc}(1) is to solve the steady-state equations, 
and then apply the 
six-dimensional Routh-Hurwitz stability criterion.  However, this approach
is intractable: the resulting Hurwitz determinants are pages-long. 

Accordingly, we take an algebraic shortcut.  
Namely, we find a parametrization 
of the set of steady states,
and then use this for the input to Routh-Hurwitz.
The result is somewhat surprising: each Hurwitz determinant except the last two (which are positive multiples of each other) is always positive.  This yields our answer to Question~\ref{q:osc}(1):
{\em For every ODE system arising from the mixed-mechanism network~\eqref{eq:mixed-network}, 
a (two-dimensional) surface in the three-dimensional space of total amounts defines the border between steady states that are stable and those that are unstable.}  
(Our result even applies to many systems 
for which the $k_i$'s are not those in Table~\ref{tab:rates};
see Proposition~\ref{prop:surface}.) 

We can now translate Question~\ref{q:osc}(2) as follows:
does the surface mentioned above consist 
of Hopf bifurcations?  We prove, using a Hopf-bifurcation criterion stated in terms of Hurwitz determinants, due to 
Yang~\cite{yang-hopf},
that the answer, at least generically, is ``yes'':
{\em 
When the unique steady state of the mixed-mechanism network~\eqref{eq:mixed-network} switches from being stable to unstable, then, generically, it undergoes a Hopf bifurcation.
}

For general one-parameter ODE systems,
there are two types of local bifurcations: 
saddle nodes (which require a zero eigenvalue of the Jacobian matrix) and 
Hopf bifurcations 
(which require a pair of pure imaginary eigenvalues of the Jacobian)~\cite{gh13}. 
We show that a saddle node bifurcation can not occur for any parameter values 
(see the proof of Proposition~\ref{prop:surface}). 
Therefore, only Hopf bifurcations are possible for the mixed-mechanism system. 

A second question we aim to answer is the following:
\begin{question} \label{q:conditions}
  Consider the mixed-mechanism network~\eqref{eq:mixed-network}.
 What conditions on the $k_i$'s guarantee
   a Hopf-bifurcation for some (positive) values of the
  total concentrations? 
\end{question}

As an answer to Question~\ref{q:conditions},
we prove that 
the catalytic constants ($k_7$ and $k_{10}$)
and association constants ($k_5$ and $k_8$)
of the distributive part of the mechanism enable oscillations to emerge
via a Hopf bifurcation.
Specifically, under the simplifying assumption that all dissociation (backward-reaction) constants are
equal ($k_2 = k_6 = k_9$), 
if the rate constants satisfy two
  inequalities -- lower bounds on $k_{10}$ and $k_5/k_8$ -- then the system generically admits a Hopf bifurcation (Proposition~\ref{prop:hopf} and Theorem~\ref{thm:main}).  
  (As a comparison, for the fully distributive dual-site network 
  described in Section~\ref{sec:related-work} below,
  the catalytic constants alone 
  enable bistability~\cite{a6maya}.)  
  Finally, we encode the relevant inequalities in a
  procedure to generate many parameter values for which we expect oscillations (Procedure~\ref{proc:rates}).

\subsection{Connection to related work} \label{sec:related-work}
Our work joins a growing number of works that harness steady-state parametrizations. 
Such results include criteria for when such parametrizations exist~\cite{johnston-param,TG}
and 
methods for using them to determine whether a network is multistationary~\cite{translated,signs,messi,MPM_MAPK}.
Going further, steady-state parametrizations
 can also be used 
to find a witness to multistationarity
or even the precise parameter regions that yield multistationarity \cite{CFMW,a6maya}.  
In this work, we use a steady-state parametrization in a novel way: to study oscillations via Hopf bifurcations.  (Our approach is similar in spirit to using Clarke's convex parameters together with a Hopf-bifurcation criterion~\cite{Domijan2009, Errami, gatermann-hopf,hadac-osc}).

As mentioned earlier, there has been much interest in the dynamics of phosphorylation systems~\cite{perspective}.  The mixed-mechanism network~\eqref{eq:mixed-network} fits into the related literature as follows.
The mixed network is a dual-site network situated
between two extremes: 
 the {\em fully processive} dual-site network -- in which the phosphorylation and dephosphorylation mechanisms are both processive -- and the {\em fully distributive} dual-site network.  One might therefore expect the dynamics of the mixed-mechanism network 
 to straddle those of the two networks.  
 This is indeed the case.  
As summarized in Table~\ref{tab:compare}, and reviewed in~\cite{perspective}, fully processive networks are globally convergent to a unique steady state~\cite{ConradiShiu,EithunShiu,Rao}, while mixed-mechanism networks admit oscillations but not bistability~\cite{SK}, and fully distributive networks admit 
bistability~\cite{bistable} (and the question of oscillations is open~\cite{perspective}).

\begin{table}[hbt]
\centering
\begin{tabular}{l c c c }
\hline
{\bf Dual-site network }                                 & {\bf Oscillations?} & {\bf Bistability?} & {\bf Global convergence?}  \\
\hline
Fully processive                                   & No            & No           & Yes                                          \\
{Mixed-mechanism}                                              & Yes           & No           & No                                           \\
Fully distributive                                 & (Open)          & Yes          & No                                           \\
\hline
\end{tabular}
\caption{Dual-site phosphorylation networks and their properties: whether they admit oscillations or bistability, and whether all trajectories converge to a unique steady state.
\label{tab:compare}}
\end{table}

Finally, we revisit 
Suwanmajo and Krishnan's
claim
mentioned earlier 
that 
the mixed-mechanism network is among the simplest enzymatic mechanisms 
with oscillations.   In support of this claim, Tung proved that the simpler system obtained from the mixed-mechanism network by taking its (two-dimensional) Michaelis-Menten approximation, is  {\em not} oscillatory \cite{ray}.   
Moreover, Rao showed that this approximation is globally convergent to a unique steady state~\cite{Rao-2}.
The validity of the Michaelis-Menten approximation for phosphorylation systems has been called into question~\cite{salazar}, and what we know about the mixed-mechanism system concurs: 
this system is oscillatory, but its Michaelis-Menten approximation is not.

The outline of our work is as follows. 
Section~\ref{sec:background} provides background on multisite phosphorylation, steady states, and Hopf bifurcations.  
Section~\ref{sec:steady-states} gives a monomial parametrization of the steady states of mixed-mechanism network. In 
Section~\ref{sec:results}, we prove our main results (described above).
We use these results in Section~\ref{sec:find-param} 
to give a procedure for generating rate constants admitting Hopf bifurcations.  
In Section~\ref{sec:simulations}, we present simulations that suggest that oscillations are the norm in the unstable-steady-state regime.
Finally, we end with a Discussion in Section~\ref{sec:discussion}.

\section{Background} \label{sec:background}
In this section, we introduce the ODEs arising from the mixed-mechanism network, and recall two criteria: the Routh-Hurwitz criterion for steady-state stability and Yang's criterion for Hopf bifurcations.

\subsection{Differential equations of the mixed-mechanism network} \label{sec:bkrd-phos}

For the mixed-mechanism network~(\ref{eq:mixed-network}), we 
let $x_1,x_2, \ldots, x_9$ denote the species
concentrations in the order given in
Table~\ref{tab:variables}. 
The dynamical system (arising from mass-action kinetics) 
defined by the mixed-mechanism network~\eqref{eq:mixed-network} is
given by the following ODEs: 
\begin{align}  
    \notag
 \dot x_1 &~=~     -k_1 x_1 x_2+k_2 x_3+k_{10} x_9 \\
    \notag
 \dot x_2 &~=~      -k_1 x_1 x_2+k_2 x_3+k_4 x_4 \\
    \notag
 \dot x_3 &~=~      k_1 x_1 x_2-(k_2+k_3) x_3 \\
    \notag
 \dot x_4 &~=~      k_3 x_3- k_4 x_4 \\
    \label{eq:OEs-mixed}
 \dot x_5 &~=~      k_4 x_4-k_5 x_5 x_6+k_6 x_7 \\
    \notag
 \dot x_6 &~=~      -k_5 x_5 x_6-k_8 x_8 x_6+(k_6+k_7) x_7+(k_9+k_{10}) x_9 \\
    \notag
 \dot x_7 &~=~      k_5 x_5 x_6-(k_6+k_7) x_7 \\
    \notag
 \dot x_8 &~=~      k_7 x_7-k_8 x_6 x_8+k_9 x_9 \\
    \notag
 \dot x_9 &~=~      k_8 x_6 x_8-(k_9+k_{10}) x_9 ~.
\end{align}
\begin{table}[hbt]
  \centering
    \begin{tabular}[ht]{|c|c|c|c|c|c|c|c|c|} \hline
      $x_1$ & $x_2$ & $x_3$ & $x_4$ & $x_5$ & $x_6$ & $x_7$ & $x_8$ & $x_9$ \\ \hline
      $S_0$ & $K$ & $S_0K$ & $S_1 K$ & $S_2$ & $P$ & $S_2 P$ & $S_1$ & $S_1 P$ \\ \hline
    \end{tabular}
  \caption{
      Assignment of variables to species for the mixed-mechanism
      network~\eqref{eq:mixed-network}. 
  }
  \label{tab:variables}
\end{table}

The conservation laws arise from the fact that the total amounts of free and bound enzyme or substrate remain constant.  
That is, as the dynamical system~\eqref{eq:OEs-mixed} progresses,
the following three conservation values, denoted by 
$\Ktot, \Ptot, \Stot \in \R_{>0}$, 
remain constant: 
\begin{align}
\label{eqn:conservation}
\notag
\Ktot ~&=~ x_2 + x_3 + x_4~, \\
\Ptot ~&=~ x_6+x_7+x_9  ~, \\
\notag
\Stot ~&=~ x_1+x_3+x_4+x_5+x_7+x_8+x_9  ~.
\end{align}

Also, a trajectory $x(t)$ beginning in 
$\R^9_{\geq 0}$ remains in $\R^9_{\geq 0}$ for all positive time $t$, so it remains in a {\em stoichiometric compatibility class}, which we
denote as follows:
\begin{align}\label{eqn:invtPoly}
\invtPoly~=~\{ x \in  \mathbb{R}^9_{\geq 0} \mid
    \text{ the conservation equations $\eqref{eqn:conservation}$ hold} \}~.
\end{align}

\subsection{Stability of steady states and the Routh-Hurwitz criterion} \label{sec:stable}

The dynamical system~\eqref{eq:OEs-mixed} arising from the mixed-mechanism network is an example of a {\em reaction kinetics system}.  That is, the system of ODEs takes the following form:
\begin{align} \label{eq:ODE}
\frac{dx}{dt} ~ = ~ \Gamma \cdot R(x)~ =:~ g(x)~,
\end{align}
where $\Gamma$ and $R$ are as follows.  Letting $s$ denote the number
of species  and $r$ the number of reactions, $\Gamma$ is an $s \times r$ matrix 
whose $k$-th column is the reaction vector of the $k$-th reaction, i.e., it encodes the net change in each species that results when that reaction takes place.
Also, $R:\Rnn^s \to \Rnn^r$  encodes the reaction rates of the $r$ reactions as functions of the $s$ species concentrations.  

A {\em steady state} (respectively, {\em positive steady state}) of a reaction kinetics system is a nonnegative concentration vector $x^* \in \Rnn^s$ (respectively, $x^* \in \mathbb{R}_{>0}^s$) at which the ODEs~\eqref{eq:ODE} vanish: $g(x^*) = 0$.  Letting $\St := \im(\Gamma)$ denote the {\em stoichiometric subspace}, 
a steady state $x^*$ is {\em nondegenerate} if ${\rm Im}\left( dg (x^*)|_{S} \right) = \St$, where $dg(x^*)$ denotes the Jacobian matrix of $g$ at $x^*$.

A nondegenerate steady state is 
locally asymptotically stable 
if each of the $\sigma:= \dim(\St)$ nonzero eigenvalues of $dg(x^*)$ has negative real part. 
Hence, a steady state is locally stable if and only if the characteristic polynomial of the Jacobian evaluated at the steady state has $\sigma$ roots with negative real part (the remaining roots will be 0).  

To check whether a polynomial has only roots with negative real parts, we appeal to the Routh-Hurwitz criterion below \cite{fg59}.

\begin{definition} \label{def:hurwitz}
The {\em $i$-th Hurwitz matrix} of a univariate polynomial 
$p(\lambda)= a_0 \lambda^n + a_{1} \lambda^{n-1} + \cdots + a_n$  
is the following $i \times i$ matrix:
\[
H_i ~=~ 
\begin{pmatrix}
a_1 & a_0 & 0 & 0 & 0 & \cdots & 0 \\
a_3 & a_2 & a_1 & a_0 & 0 & \cdots & 0 \\
\vdots & \vdots & \vdots &\vdots & \vdots &  & \vdots \\
a_{2i-1} & a_{2i-2} & a_{2i-3} & a_{2i-4} &a_{2i-5} &\cdots & a_i
\end{pmatrix}~,
\]
in which the $(k,l)$-th entry is $a_{2k-l}$ 
  as long as $0\leq 2 k - l \leq n$, and
  $0$ otherwise. 
\end{definition}

\begin{proposition}[Routh-Hurwitz criterion] \label{prop:routh-hurwitz}  A 
polynomial 
$p(\lambda)= a_0 \lambda^n + a_{1} \lambda^{n-1} + \cdots + a_n$  
with $a_0 > 0$
has all roots with negative real part if and only if all $n$ of its Hurwitz matrices have positive determinant ($\det H_i >0$ for all $i=1,\dots, n$).
\end{proposition}

\subsection{Hopf bifurcations and a criterion due to Yang} \label{sec:hopf}
A {\em simple Hopf bifurcation} is a bifurcation in which a single complex-conjugate pair of eigenvalues of the Jacobian matrix crosses the imaginary axis, while all other eigenvalues remain with negative real parts.  Such a bifurcation, if it is supercritical, generates nearby {\em oscillations} or periodic orbits~\cite{liu}.

To detect simple Hopf bifurcations, we will use a criterion of Yang that characterizes Hopf bifurcations in terms of Hurwitz-matrix determinants (Proposition~\ref{prop:yang}).

\noindent {\bf Setup for Yang's criterion.}
We consider an ODE system parametrized by $\mu \in \mathbb{R}$:
    \begin{align*}
        \dot x ~=~ g_{\mu}(x)~,
    \end{align*}
where $x \in \mathbb{R}^n$, and $g_{\mu}(x)$ varies smoothly in $\mu$ and $x$.  Assume that $x_0 \in \mathbb{R}^n$ is a steady state of the system defined by $\mu_0$, that is, $g_{\mu_0}(x_0)=0$.  Assume, furthermore, that we have a smooth curve of steady states:
    \begin{align}  \label{eq:curve}
    \mu ~\mapsto ~ x(\mu)~
    \end{align}
(that is, $g_{\mu}\left( x(\mu) \right)= 0$ for all $\mu$) and that $x(\mu_0)=x_0$.  
Denote the characteristic polynomial of the Jacobian matrix of $g_{\mu}$, evaluated at $x(\mu)$, as follows:
    \begin{align*}
         p_{\mu}(\lambda) 
         ~:=~ \det \left( \lambda I - {\rm Jac}~ g_{\mu} \right)|_{x = x(\mu)}
         ~=~ \lambda^n + a_{1}(\mu) \lambda^{n-1} + \cdots + a_n(\mu)~, 
    \end{align*}
and, for $i=1,\dots, n$, let $H_i(\mu)$ denote the $i$-th Hurwitz matrix of $p_{\mu}(\lambda)$.

\begin{proposition} [Yang's criterion~\cite{yang-hopf}] \label{prop:yang}
  Assume the above setup.  Then, there is a simple Hopf bifurcation at
  $x_0$ with respect to $\mu$ if and only if the following hold:
    \begin{enumerate}[(i)]
        \item $a_n(\mu_0)>0$,
        \item $\det H_1(\mu_0)>0$, $\det H_2(\mu_0)>0$, \dots, $\det
          H_{n-2}(\mu_0)>0$, and 
        \item $\det H_{n-1}(\mu_0)= 0$ and $\frac{d( \det H_{n-1}(\mu)
            )}{d \mu}|_{\mu = \mu_0} \neq 0$.
    \end{enumerate}
\end{proposition}

\begin{remark} \label{rmk:yang-vs-liu}
  Liu~\cite{liu} gave an earlier version of Yang's Hopf-bifurcation
  criterion (Proposition~\ref{prop:yang}), using a variant of the
  Hurwitz matrices that differs from ours.
\end{remark}

\section{Steady states of the mixed-mechanism network} \label{sec:steady-states}
In this section, we recall that the mixed-mechanism network admits a unique steady state in each compatibility class (Proposition~\ref{prop:unique-steady-state}), and prove that the set of steady states admits a monomial parametrization (Theorem~\ref{thm:mon}).
We then use this parametrization to analyze the space of 
compatibility classes (Proposition~\ref{prop:param-comp-classes}).

\subsection{Uniqueness of steady states}

Suwanmajo and Krishnan proved that, for {\em every} choice of 
positive rate constants 
and positive total amounts, 
the mixed-mechanism network 
does {\em not} admit multiple positive steady states~\cite[\S A.2]{SK}.  Additionally, there are no boundary steady states in any compatibility class $\invtPoly$, as in~\eqref{eqn:invtPoly}, and $\invtPoly$ is compact.  
Hence,  via a standard application of the Brouwer fixed-point theorem~(e.g., \cite[Remark 3.9]{TSS}), 
 there is always a unique steady state:
  
\begin{proposition}[Uniqueness of steady states] \label{prop:unique-steady-state}
For any choice of 
positive rate constants $k_i$ and positive total amounts $\Ktot,$ $\Ptot,$ and $\Stot$, the dynamical system~\eqref{eq:OEs-mixed} arising from the mixed-mechanism network 
has a unique steady state in $\invtPoly$, and it is a positive steady state.
\end{proposition}  
Proposition~\ref{prop:unique-steady-state} precludes the existence of multiple positive steady states, and hence the existence of a saddle node bifurcation. Thus, a Hopf bifurcation is the only other one-parameter bifurcation which may occur. Indeed, we will show that a Hopf bifurcation exists for some parameter values in Section~\ref{sec:results}. 

Also, Proposition~\ref{prop:unique-steady-state} proves part of a conjecture that we posed~\cite{ConradiShiu}.  
The other half of the conjecture, however, posited that mixed-mechanism systems, 
like fully processive systems~\cite{ConradiShiu,EithunShiu}, 
are globally convergent to the unique steady state.  Suwanmajo and Krishnan demonstrated that this is false: the system can exhibit oscillatory behavior~\cite{SK}!  

This capacity for oscillations is the focus of this work, and our analysis will harness a monomial parametrization of the steady states.  We turn to this topic now.

\subsection{A monomial parametrization of the steady states}

The steady states of the mixed-mechanism network can be parametrized by monomials (and thus is said to have ``toric steady states''~\cite{TSS}):
\begin{proposition}[Parametrization of the steady states] \label{thm:mon}
For every choice of rate constants $k_i >0$,
the set of positive steady states of the mixed-mechanism system~\eqref{eq:OEs-mixed} is three-dimensional and is
  the image of the following map
  $\chi = \chi_{k_1, \dots, k_{10} }$: 
  \begin{align} \label{eq:param}
    \chi : \mathbb{R}^3_{+} ~& \to~ \mathbb{R}^{9}_{+} \\
    (x_1, x_2, x_6) ~&\mapsto ~(x_1, x_2, \dots, x_{9})~, \notag
  \end{align}  
  given by
    \begin{align*}
    x_3  ~&:=~  
    \frac{{k_1}  }{{k_2} +{k_3} } {x_1} {x_2}, 
    \quad  \quad 
    x_4 ~:=~ 
    \frac{{k_1} {k_3} }{ ({k_2}+k_3) {k_4} } {x_1} {x_2},    
    \quad \quad
    x_5  ~:=~
    \frac{{k_1} {k_3}  ({k_6}+{k_7}) }{  ({k_2} +{k_3} ) {k_5} {k_7} } \frac{{x_1} {x_2} }{x_6},  \\
    x_7  ~&:=~ 
    \frac{{k_1} {k_3}  }{ ({k_2}+{k_3}) k_7 } {x_1} {x_2},    
    \quad 
    x_8  ~:=~ 
    \frac{{k_1} {k_3}  ({k_9}+{k_{10}}) }{ ({k_2} +{k_3} ) {k_8} {k_{10}} } \frac{ {x_1} {x_2}}{x_6},     
    \quad 
    x_9  ~:=~ 
   \frac{{k_1} {k_3} }{ ( {k_2} +{k_3})  {k_{10}} }{x_1} {x_2}~.
    \end{align*}
\end{proposition}

\begin{proof}
It is straightforward to check that the image of 
$\chi$ is contained in the set of steady states: 
after substituting $\chi(x_1,x_2,x_3)$, 
the right-hand side of the mixed-mechanism network ODEs~\eqref{eq:OEs-mixed} vanishes.  Conversely, let $x^*=(x_1,x_2,\dots, x_9)$ be a positive steady state.
The right-hand side of the ODEs~\eqref{eq:OEs-mixed} vanish at $x^*$, so, in the following order, we 
use $\dot x_3=0$ to solve for $x_3$ in terms of $x_1$ and $x_2$, 
use $\dot x_4=0$ to solve for $x_4$ via $x_3$ which was already obtained, 
use $\dot x_1=0$ to obtain $x_9$,
use $\dot x_9=0$ to obtain $x_8$,
use $\dot x_8=0$ to obtain $x_7$, and finally 
use $\dot x_7=0$ to obtain $x_5$.
This yields precisely the parametrization~\eqref{eq:param}, so $x^*$ is in the image of $\chi$.
\end{proof}

\begin{remark}
The parametrization~\eqref{eq:param} appeared earlier in~\cite{perspective}.
\end{remark}

\begin{remark}
That we could achieve a steady-state parametrization was expected, due to Thomson and Gunawardena's rational parametrization theorem for multisite systems ~\cite{TG}.  
\end{remark}

\begin{remark}
In the parametrization $\chi$ in Theorem~\ref{thm:mon}, we divide by $x_6$, so $\chi$ is technically not a monomial map.  However, $\chi$ can be made monomial: we introduce $y:=\frac{x_1}{x_6}$, so that the parametrization accepts as input $(y,x_2,x_6)$, and then 
$x_1$ is replaced by $y x_6$.
\end{remark}

\subsection{A parametrization of the compatibility classes}
Every compatibility class 
$\invtPoly$
of the mixed-mechanism network, by definition~\eqref{eqn:invtPoly}, is uniquely determined by a choice of total amounts
$(\Ktot,~\Ptot,~\Stot) \in \mathbb{R}^3_{>0}$. Thus, we identify the set of compatibility classes with $\{(\Ktot,~\Ptot,~\Stot)\}= \mathbb{R}^3_{>0}$.  
We parametrize this set below (Proposition~\ref{prop:param-comp-classes}).

Let $\phi: \mathbb{R}_{> 0}^9 \to \mathbb{R}^3_{>0}$ denote the map sending a vector of concentrations to the corresponding total amounts $( \Ktot,~\Ptot,~\Stot)$, as in~\eqref{eqn:conservation}:
 \begin{align} \label{eq:phi}
 	\phi(x)~:=~ ( x_2 + x_3 + x_4~, ~
	 x_6+x_7+x_9  ~,~ x_1+x_3+x_4+x_5+x_7+x_8+x_9 ) ~.
 \end{align}
Each compatibility class $\invtPoly$ contains a unique positive steady state (Proposition~\ref{prop:unique-steady-state}), and the positive steady states are  parametrized by  $\chi$ 
from Theorem~\ref{thm:mon}, so 
the space of compatibility classes is parametrized as follows:

\begin{proposition}[Parametrization of the compatibility classes] 
\label{prop:param-comp-classes}
Identify every compatibility class $\invtPoly$ of the mixed-mechanism network~\eqref{eq:mixed-network}, with the corresponding total amounts 
$(\Ktot,~\Ptot,~\Stot) \in \mathbb{R}^3_{>0}$. 
Then, for every choice of positive rate constants $k_i$,
the following is a bijection that sends a vector $ (x_1,x_2,x_6) \in \mathbb{R}^3_{>0} $
to the compatibility class in which the unique steady state is $\chi(x_1,x_2,x_6) $:
\begin{align*}
\phi \circ \chi: ~
    \mathbb{R}^3_{>0} \to \mathbb{R}^3_{>0} = \{(\Ktot,~\Ptot,~\Stot)\}~,
\end{align*}
where 
	$\phi$ is as in~\eqref{eq:phi} and 
	$\chi$ is the steady-state parametrization~\eqref{eq:param}.
The map $\phi \circ \chi$ 
  is given by
  \begin{displaymath}
    \begin{split}
      (x_1,x_2,x_6)  &\mapsto 
      \Biggl(
      x_{2}+ \frac{k_{1}}{k_{2}+k_{3}}  \left(1 + \frac{k_{3}}{k_{4}}\right) x_{1} x_{2}
      , \quad  
      x_6 + \frac{k_{1} k_3}{k_{2}+k_{3}} 
      \left(\frac{1}{k_{7}} + \frac{1}{k_{10}}\right) x_{1} x_{2}, \\  
      &\qquad 
      x_{1} + \frac{k_{1} k_3}{k_{2} + k_{3}}  \left[
        \left( \frac{1}{k_3} + \frac{1}{k_{4}} + \frac{1}{k_{7}} + \frac{1}{k_{10}} \right) 
        + \frac{1}{x_{6}}
        \left(\frac{k_{6}+k_{7}}{ k_{5} k_{7}} +
          \frac{k_{10}+k_{9}}{k_{10} k_{8}} \right) \right] x_{1} x_{2}  \Biggr)~,
    \end{split}
  \end{displaymath}
  which becomes, when the rate constants are those in Table~\ref{tab:rates}, the following:
\begin{align} \label{eq:param-c-class}
	 (x_1,~x_2,~x_6) ~ & \mapsto ~ \left(x_1 x_2 + x_2,~
	 x_6+\frac{1009}{1800}x_1 x_2, ~ 
	 x_1+\frac{2809}{1800}x_1 x_2 + \frac{161}{900} \frac{x_1x_2}{x_6} \right)~. 
	\end{align}
\end{proposition}

\begin{example} \label{ex:param}
Consider the mixed-mechanism system with 
rate constants from Table~\ref{tab:rates}.
To compute the unique steady state $x^*$ in the
 compatibility class given by
   $ (\Ktot,~ 
	\Ptot,~
	\Stot) =(17.5,~5,~40)$,
we use Proposition~\ref{prop:param-comp-classes}.
Namely, we know that $\phi \circ \chi (x_1^*,x_2^*,x_6^*) = (17.5,5,40)$,
so we solve 
(using, e.g., {\tt Mathematica}~\cite{mathematica}) for the unique positive solution:
\[ 
(x_1^*,~x_2^*,~x_6^*) ~ \approx ~ (1.0134, ~8.6916, ~0.0624)~.
\]  
We obtain the remaining coordinates of $x^*$ using the parametrization
$\chi$ in~\eqref{eq:param}:
\begin{align} \label{eq:steady-ste-for-17-5-40}
	x^* ~
	&= ~
	\chi (x_1^*,~x_2^*,~x_6^*) \\ 
	~&\approx ~
	(1.0134,~ 8.6916,~ 4.4041,~ 4.4041,~ 1.4893,~ 0.0624,~4.8935,~23.7512,~ 0.0440 )~. \notag
\end{align}
\end{example}

\subsection{Steady states and Hopf bifurcations} \label{sec:hopf-figure}
Our analysis of oscillations in the mixed-mechanism system is based on
Hopf bifurcations.  
Hopf-bifurcation diagrams are displayed in 
Figure~\ref{fig:1-para-conti}, 
where the total amounts are the bifurcation parameters
(c.f.\ Figure~\ref{fig:hopf} which is with respect to $\Ktot$).
Figure~\ref{fig:1-para-conti}
suggests that, in the 3-dimensional space of 
total amounts, there is a surface of Hopf bifurcations.  Indeed, 
we will see in the next section that this is the case
(see Theorem~\ref{thm:main} and Figure~\ref{fig:2-para-conti}).

\begin{figure}[h]

  \begin{subfigure}{0.3\textwidth}
    \includegraphics[width=\textwidth]{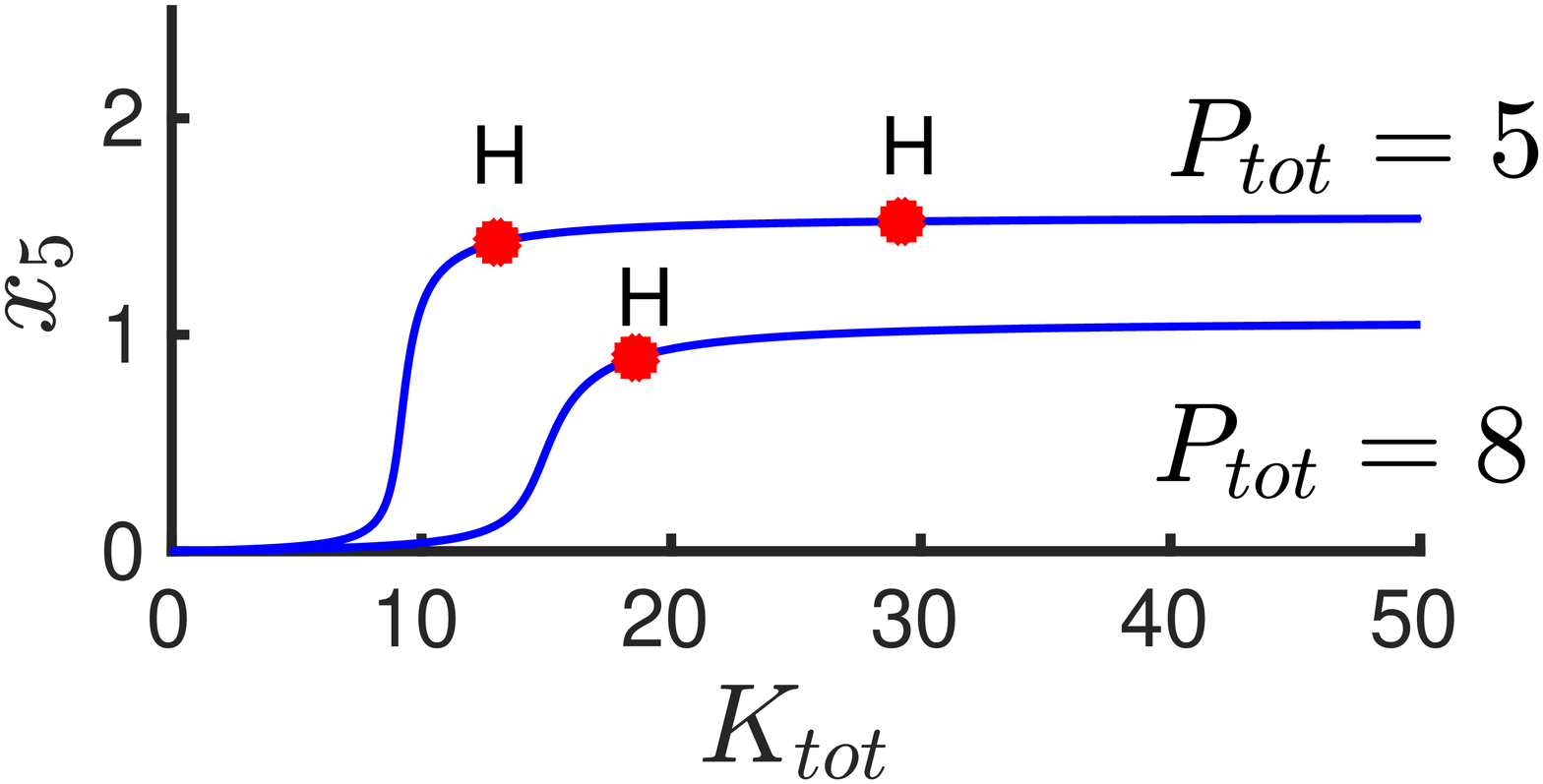}
    \subcaption{
      Bif.~parameter $\Ktot$.
    }
  \end{subfigure}
  \hfill
  \begin{subfigure}{0.3\textwidth}
    \includegraphics[width=\textwidth]{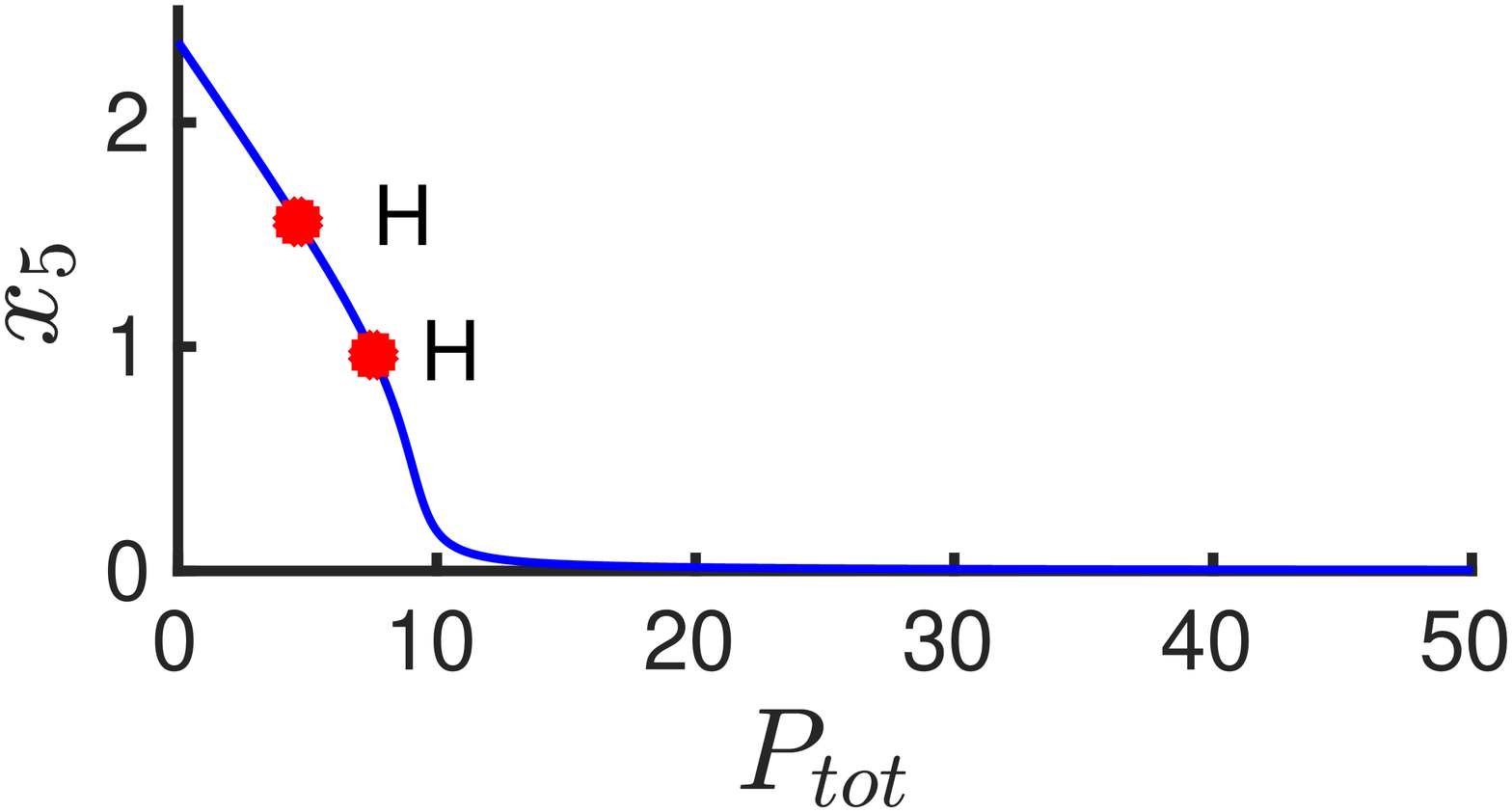}
    \subcaption{
      Bif.~parameter $\Ptot$.
    }    
  \end{subfigure}
  \hfill
  \begin{subfigure}{0.3\textwidth}
    \includegraphics[width=0.9\textwidth]{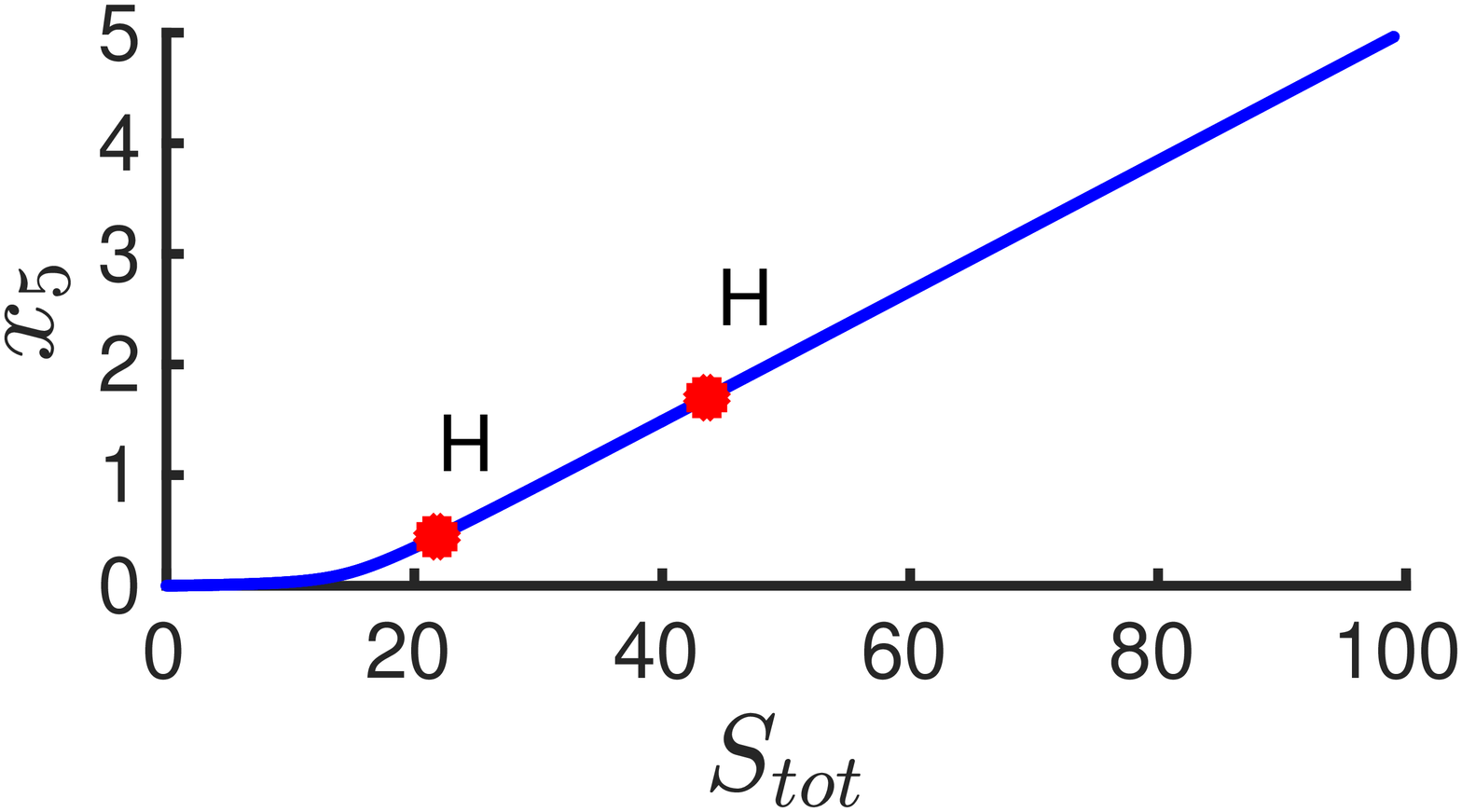}
    \subcaption{
      Bif.~parameter $\Stot$. 
    }    
  \end{subfigure}
  \caption{
    \label{fig:1-para-conti}
    Numerical continuation of the unique positive steady state,
     in~\eqref{eq:steady-ste-for-17-5-40}, when
     $ (\Ktot,~ \Ptot,~	\Stot) =(17.5,~5,~40)$:
    (a) For $\Ptot =  5,8$ and $\Stot=40$, we observe
    (supercritical) Hopf
    bifurcations at $\Ktot \approx 13.0296$, $29.2251$ ($\Ptot=5$) and
    $\Ktot \approx 18.5758$ ($\Ptot =8$).
    (b) For $\Ktot = 5$ and $\Stot=40$, we observe  (supercritical)
    Hopf
    bifurcations at $\Ptot \approx 4.6310$ and $\Ptot \approx 7.5479$.
    (c) For $\Ktot = 17.5$ and $\Ptot=5$, we observe
    (supercritical) Hopf bifurcations at $\Stot \approx 21.8213$ and
    $\Stot \approx 43.5944$.  All figures in this work were made 
    using {\tt Matcont}~\cite{matcont}.
  }
\end{figure}

\section{Hopf bifurcations in the mixed-mechanism system} \label{sec:results}
We saw in the previous section that the mixed-mechanism network yields a unique positive steady state in each compatibility class. 
Now we show that the compatibility classes with 
a  {\em stable} steady state are separated from those with an {\em unstable} steady state 
by a single surface~$\mathcal H$ (Proposition~\ref{prop:surface} and Theorem~\ref{thm:hopf}), 
and, under stronger hypotheses, crossing the surface $\mathcal H$ generically corresponds to undergoing a Hopf bifurcation (Theorem~\ref{thm:main}).
(Recall that {\em generically} means that the exceptional set has zero measure.  So, we will show that the subset of the surface corresponding to non-Hopf points has dimension at most 1.)

To simplify computations, we assume that dissociation (backward-reaction) constants are
equal: $k_2 = k_6 = k_9$.    In chemistry, the forward reaction is usually more thermodynamically favorable than the backward  reaction. Therefore, the rate constant of a forward reaction is much larger than the rate constant of the backward reaction \cite{adp18}.  We choose small values for  the dissociation rate constants in Section~5, similar to what was done in \cite{fha14}.

\begin{proposition} \label{prop:surface}
  Consider the dynamical system~\eqref{eq:OEs-mixed} arising from the mixed-mechanism network and any positive rate constants for which $k_2 = k_6 = k_9$.
  Then:
\begin{enumerate}
    \item Every compatibility class $\invtPoly$ contains a unique (positive) steady state $x^*$. 
    \item Exactly one of the following holds:
   \begin{enumerate}[(a)]
	\item The unique steady state $x^*$ in each compatibility class $\invtPoly$ is locally asymptotically stable.	
	\item  In the space of total amounts $\{(\Ktot,\Ptot,\Stot)\} = \mathbb{R}^3_{>0}$, which we identify with the space of compatibility classes $\invtPoly$, 
    a surface $\mathcal{H}$ defines the border between those $\invtPoly$ whose unique steady state $x^*$ is locally asymptotically stable and those $\invtPoly$ for which $x^*$ is unstable.
	\end{enumerate}
\end{enumerate}
\end{proposition}

\begin{proof}
Item 1 follows from Proposition~\ref{prop:unique-steady-state}.

For item 2, 
let $J$ denote the Jacobian matrix of the mixed-mechanism system~\eqref{eq:OEs-mixed},
with equal 
 dissociation constants: $k_2 = k_6 = k_9=:k_b$, 
evaluated at the parametrized steady state $\chi(x_1,x_2,x_6)$, from~\eqref{eq:param}.  The characteristic polynomial of $J$ is:
\begin{align*}
    p(\lambda) ~:=~  \det (\lambda I - J) 
    ~=~ \lambda^3 (\lambda^6+ b_1 \lambda^5 + b_2 \lambda^4 + \cdots + b_6 )~,
\end{align*}
where the coefficients $b_i$ (displayed below) are rational functions in 
$x_1,x_2,x_6$ and the $k_i$'s.
  To streamline reading we only give the complete numerator
  of $b_6$ and $b_1$. 
  The full coefficients can be found in the {\tt Mathematica} file
  \texttt{mixed\_coeffs\_charpoly\_kb.nb}\footnote{This file and others mentioned below are in the Supporting Information; see Appendix~\ref{sec:A}.}. 

{\footnotesize
\begin{align} \label{eq:b6}
 \text{numerator} (b_6) ~&=~ k_{1}^2 k_{3}^2 k_{4} (k_{10}+k_{7})
                           (k_{10} k_{5} k_{7}+k_{5} k_{7} k_b+k_{10}
                           k_{8} (k_{7}+k_b)) x_{1} x_{2}^2 \\
                         &\qquad 
                           \notag
                           +k_{1} k_{10} k_{3} k_{4} k_{7} (k_{3}+k_b) (k_{10} k_{5} k_{7}+k_{5} k_{7}
                           k_b+k_{10} k_{8} (k_{7}+k_b)) x_{2} x_{6} \\
                         &\qquad 
                           \notag
                           +k_{10}^2 k_{4} k_{5} k_{7}^2 k_{8} (k_{3}+k_b)^2
                           x_{6}^2+k_{1} k_{10}^2 (k_{3}+k_{4}) k_{5}
                           k_{7}^2 k_{8} (k_{3}+k_b) x_{1} x_{6}^2 \\
                         &\qquad 
                           \notag
                           +k_{1} k_{10}
                           k_{5} k_{7} (k_{10} k_{4} k_{7}+k_{3} k_{4}
                           k_{7}+k_{10} k_{3} (k_{4}+k_{7})) k_{8}
                           (k_{3}+k_b) x_{2} x_{6}^2 
 \end{align}
\begin{align}
  \notag
  \text{numerator}(b_5) ~&=~ k_{1}^2 k_{3}^2 k_{4} (k_{10}+k_{7})
                           (k_{10}+k_b) (k_{7}+k_b)  x_{1} x_{2}^2 \\
                           &\qquad 
                             \notag
                             +k_{1} k_{10}
                             k_{3} k_{4} k_{7} (k_{10}+k_b) (k_{3}+k_b)
                             (k_{7}+k_b) x_{2} x_{6} + \ldots
  \\
  \notag
  \text{numerator}(b_4) ~&=~ 
                           k_{1} k_{3} k_{4} (k_{10}+k_{7})
                           (k_{10}+k_b) (k_{3}+k_b) (k_{7}+k_b) x_{1}
                           x_{2} + \ldots 
  \\
  \notag
  \text{numerator}(b_3) ~&=~ 
                           \ldots 
                           + k_{1}^2 k_{3} \Bigl(k_{10}^2 (k_{7}+k_b)+k_{7} k_b
                           (k_{3}+k_{4}+k_{7}+k_b) \\ \notag
                         &\qquad \qquad 
                           +k_{10}
                           \left((k_{7}+k_b)^2+k_{3} (2 k_{7}+k_b)+k_{4} (2
                           k_{7}+k_b)\right)\Bigr) x_{1}^2 x_{2} +
                           \ldots 
  \\
  \notag
  \text{numerator}(b_2) ~&=~ 
                           \ldots +k_{1}^2 k_{3} (k_{7} k_b+k_{10} (2
                           k_{7}+k_b)) x_{1}^2 x_{2} + \ldots 
  \\
  \notag
  \text{numerator}(b_1) ~&=~ 
                           k_{1} k_{3} (k_{7} k_b+k_{10} (2 k_{7}+k_b)) x_{1}
                           x_{2} 
                           +k_{10} k_{7} (k_{3}+k_b)
                           (k_{10}+k_{3}+k_{4}+k_{7}+3 k_b) x_{6} \\ \notag
                         &\qquad 
                           +k_{1} k_{10} k_{7} (k_{3}+k_b) x_{1} x_{6}
                           +k_{1} k_{10} k_{7} (k_{3}+k_b) x_{2} x_{6} 
                           +k_{10} k_{7} (k_{5}+k_{8}) (k_{3}+k_b) x_{6}^2 
  \\
  \intertext{\normalsize And for the denominators:}
  \notag
  \text{denominator}(b_6) ~&=~ k_{10} (k_{b}+k_{3}) k_{7}
  \\
  \notag
  \text{denominator}(b_i) ~&=~ k_{10} (k_{b}+k_{3}) k_{7} x_{6}~, \quad \text{for }
                             i=2,3,4, 5~.
\end{align}
} 

As $x_1$, $x_2$, $x_6$ and the $k_i$ are positive,
thus $b_1, b_2, \dots, b_6>0$ 
(in the aforementioned {\tt Mathematica} file, we checked 
the above numerators are sums of only positive monomials).

Recall that, due to the 3 conservation laws~\eqref{eqn:conservation}, the Jacobian matrix has rank 6, not 9.
Accordingly, the relevant Hurwitz matrix, 
namely, for $p(\lambda)/\lambda^3$,
is as follows:
\begin{align*}
\begin{pmatrix}
b_1&  1&  0&  0&  0&  0\\
b_3&  b_2&  b_1&  1&  0&  0 \\
b_5&  b_4&  b_3&  b_2&  b_1&  1 \\
0&  b_6&  b_5&  b_4&  b_3&  b_2 \\
0&  0&  0&  b_6&  b_5&  b_4 \\
0&  0&  0&  0&  0&  b_6 
\end{pmatrix}    
\end{align*}

Consider the Hurwitz determinants.  First $\det H_1 = b_1 >0$. 
The next 3 Hurwitz determinants 
are also positive:
{
  \footnotesize
  \begin{align*}
    \text{numerator} (    \det H_2) ~&=~ 
                                       k_{1}^3 k_{3}^2 (k_{7}
                                       k_b+k_{10} (2 k_{7}+k_b))^2
                                       x_{1}^3 x_{2}^2 \\
                                     &\qquad 
                                       +k_{1}^3 k_{10}
                                       k_{3} k_{7} (k_{3}+k_b) (k_{7}
                                       k_b+k_{10} (2 k_{7}+k_b)) x_{1}^3
                                       x_{2} x_{6} + \ldots \\
    \text{numerator} (    \det H_3) ~&=~ 
                                       k_{1}^5 k_{3}^3 (k_{10} k_{5}
                                       k_{7}+k_{5} k_{7} k_b+k_{10}
                                       k_{8} (k_{7}+k_b)) (k_{7}
                                       k_b+k_{10} (2 k_{7}+k_b))^2
                                       x_{1}^5 x_{2}^3 x_{6}  + \ldots\\
    \text{numerator} (    \det H_4) ~&=~  
                                       k_{1}^7 k_{3}^4 (k_{10} k_{5}
                                       k_{7}+k_{5} k_{7} k_b+k_{10}
                                       k_{8} (k_{7}+k_b)) (k_{7}
                                       k_b+k_{10} 
                                       (2 k_{7}+k_b))^2 \\
                                     &\qquad 
                                       \Bigl(k_{5}
                                       k_{7} (k_{3}+k_{4}+k_{7})
                                       k_b+k_{10}^2 k_{8}
                                       (k_{7}+k_b) + \\
                                     &\qquad \qquad 
                                       k_{10} 
                                       (k_{3}+k_{4}+k_{7}) (k_{5}
                                       k_{7}+k_{8} (k_{7}+k_b))\Bigr)
                                       x_{1}^7 x_{2}^4 x_{6}^2 +
                                       \ldots \\
    \intertext{\normalsize where the denominators, which are positive, are, respectively:}
    \text{denominator} (\det H_2) ~&=~ k_{10}^2 k_{7}^2 (k_{b}+k_{3})^2 x_{6}^2 \\
    \text{denominator} (\det H_3) ~&=~ k_{10}^3 k_{7}^3 (k_{b}+k_{3})^3 x_{6}^3\\
    \text{denominator} (\det H_4) ~&=~ k_{10}^4 k_{7}^4 (k_{b}+k_{3})^4 x_{6}^4 
  \end{align*}
}
(We display only the leading terms of the
  polynomials; the complete polynomials together with an algorithmic
  verification of positivity are in
  \texttt{mixed\_Hi.nb}.)
The final Hurwitz determinant is $\det H_6 = (b_6) (\det H_5)$, and we saw that $b_6 >0$.  
So, by the Routh-Hurwitz criterion (Proposition~\ref{prop:routh-hurwitz}), the steady state $\chi(x_1,x_2,x_6)$ is locally stable if and only if  $\det H_5 >0$.  

Hence, the surface $\mathcal{H}$ that delineates the boundary between compatibility classes with stable steady states vs. those with unstable steady states is defined by
 $\det H_5 \circ  (\phi\circ \chi)^{-1}=0$, where $\phi\circ \chi$ is the parametrization of compatibility classes from Proposition~\ref{prop:param-comp-classes}.
If $\mathcal H$ intersects the positive orthant $\mathbb{R}^3_{>0}$, then case (b) of the proposition holds.  
Otherwise, if $\mathcal H \cap \mathbb{R}^3_{>0} = \emptyset$, then
we claim that we are in case (a).  To show this, we need to verify that $\det H_5(x_1,x_2,x_6)>0$ for some $(x_1,x_2,x_6) \in \mathbb{R}^3_{>0}$.  
The denominator of 
$\det H_5(x_1,x_2,x_6)$ is strictly positive:
\begin{displaymath}
  \text{denominator} (\det H_5) =  k_{10}^5 k_{7}^5 (k_{3}+k_b)^5 x_{6}^5.
\end{displaymath}
So we need only show that the numerator of 
$\det H_5(x_1,x_2,x_6)$ is strictly positive for some $(x_1,x_2,x_6) \in \mathbb{R}^3_{>0}$.

To this end, we view this numerator as a polynomial in $x_1$ (so the
coefficients are rational functions of $x_2$, $x_6$, and the
$k_i$'s):
{
  \footnotesize
  \begin{align}
    \notag
      \text{numerator} (\det H_5) ~&=~ x_{1}^9 x_{2}^4
      \Biggl(
      \frac{k_{10} k_{7} x_{6} (k_{3}+kb)}{k_{3} (k_{10} (2
        k_{7}+kb)+k_{7} kb)}+x_{2}
      \Biggr) \\
      \label{eq:numerator-x1}
      &\qquad
      \Biggl[
      k_{8} x_{6}
      \left(\alpha_{01}+\alpha_{10}\frac{k_{5}}{k_{8}}\right)
      +k_{8}^2 x_{6}^2
      \left(
        \alpha_{02}+\alpha_{11} \frac{k_{5}}{k_{8}}+\alpha_{20}
        \left(\frac{k_{5}}{k_{8}}\right)^2
      \right) +
      \\ \notag
      &\qquad \qquad
      k_{8}^3 x_{6}^3 
      \left(
        \alpha_{03}+\alpha_{12}\frac{k_{5}}{k_{8}}
        +\alpha_{21}
        \left(\frac{k_{5}}{k_{8}}\right)^2
        +\alpha_{30}
        \left(\frac{k_{5}}{k_{8}}\right)^3
      \right)
      \Biggr]
      + \text{lower degree terms in} \; x_1~,
  \end{align}
}
where the coefficients $\alpha_{ij}$ are sums of (many) positive
  monomials and are given in the file \texttt{mixed\_analyis\_H5N\_x1\_LT.nb}.
Therefore (for fixed $x_2$ and $x_6$) when $x_1$ is sufficiently large, the expression~\eqref{eq:numerator-x1} is positive, as desired.
\end{proof}

The proof of Proposition~\ref{prop:surface} 
focused on the surface $\mathcal H$ defined by the equation  $\det H_5 \circ  (\phi\circ \chi)^{-1}=0$.  This surface sometimes meets the positive orthant $\mathbb{R}^3_{>0}$, 
and indeed we show that this is the case when certain relationships hold among the rate constants.

\begin{theorem} \label{thm:hopf}
Consider the dynamical system~\eqref{eq:OEs-mixed} arising from the mixed-mechanism network.
Assume the positive rate constants satisfy $k_2 = k_6 = k_9$ and the following inequality:
\begin{equation}
  \label{eq:kc_condi}
  k_{10} k_{3} k_{4}-(k_{3}+k_{4}) (k_{3}+k_{7}) (k_{4}+k_{7}) ~>~ 0~.
\end{equation}
If $k_5/k_8$ is sufficiently large, then there is 
a compatibility class $\invtPoly$ 
whose unique steady state $x^*$ is unstable.
\end{theorem}

\begin{proof}
Assume that the rate constants satisfy $k_2 = k_6 = k_9=:k_b$ and~\eqref{eq:kc_condi}.
By the proof of Proposition~\ref{prop:surface},
 a steady state $\chi(x_1,x_2,x_6)$
of the mixed-mechanism system~\eqref{eq:OEs-mixed} 
  is locally stable if and only if  $\det H_5(x_1,x_2,x_6) >0$.  
  We also saw in that proof that the denominator of $\det H_5(x_1,x_2,x_6)$ is strictly positive for all 
  $(x_1,x_2,x_6) \in \mathbb{R}^3_{>0}$.  
  So, by Proposition~\ref{prop:routh-hurwitz}, it suffices to show that if $k_5/k_8$ is sufficiently large,
  then there exists $(x^*_1,x^*_2,x^*_6) \in \mathbb{R}^3_{>0}$ such that
   the numerator of 
   $\det H_5(x^*_1,x^*_2,x^*_6)$ is strictly negative: this would show that
   the steady state $x^*:=\chi(x^*_1,x^*_2,x^*_6)$ is unstable.
  
  To this end, 
  view the numerator of $\det H_5$ as a polynomial in $x_2$
with coefficients in $x_1$, $x_6$, and the $k_i$'s. It is a degree-$9$  
polynomial in $x_2$ of the following form (see the file \texttt{mixed\_analysis\_H5N\_x2\_LT.nb}): 
\begin{align}\label{eq:h5}
 {\rm numerator} (\det H_5) ~&=~ 
  	k_1^9\left( \alpha_0 x_6^3 + \alpha_1 x_6^2 + \alpha_2 x_6 +
    \alpha_3 \right) 
    \left( x_1^5 + 
    \frac{k_{10} k_{7} (k_{3}+k_{b})}{k_{3} (k_{10} (2
    k_{7}+k_{b})+k_{7} k_{b})} 
    x_1^4 x_6 \right)  x_2^9 
  \notag
  \\  
      & \quad + \text{lower degree terms}~,
\end{align}
where $\alpha_0$, \ldots, $\alpha_3$ 
are rational functions in $k_b, k_3, k_4, k_5, k_7, k_8, k_{10}$.  
These functions $\alpha_i$ are given in \texttt{mixed\_analysis\_H5N\_x2\_LT.nb}.

We now analyze $\alpha_0$, which has the following form (see  \texttt{mixed\_analysis\_H5N\_x2\_LT.nb}):
\begin{align} \label{eq:alpha-0}
  \alpha_0 = k_8^3
  	\left(\beta_{0} \left(\frac{k_5}{k_8}\right)^3 +
    \beta_{1} \left(\frac{k_5}{k_8}\right)^2 + \beta_{2}
    \left(\frac{k_5}{k_8}\right) + \beta_{3} \right)~,
\end{align}
where each coefficient $\beta_i$ is a rational function in $k_b, k_3, k_4, k_7, k_{10}$ (and hence does not depend on $k_1$, $k_5$, or $k_8$).  In particular, $\beta_0$ is the following polynomial: 
\begin{align*} 
  \beta_{0} &~=~ -k_{1}^9 k_{3}^5 k_{7}^3 ~
  	(k_{10} k_{3} k_{4}-(k_{3}+k_{4}) (k_{3}+k_{7}) (k_{4}+k_{7}))~
              (k_{10}+k_{b})^3 ~ (k_{7} k_{b}+k_{10} (2 k_{7}+k_{b}))^2~.
            \end{align*}               
It follows that $\beta_0<0$ when inequality~\eqref{eq:kc_condi} holds.

Thus, when~\eqref{eq:kc_condi} holds, then, by equation~\eqref{eq:alpha-0}, 
the inequality 
$\alpha_0<0$ holds for $k_5/k_8$ sufficiently large. 
In this case, the cubic polynomial in $x_6$ appearing in~\eqref{eq:h5}, and hence also the coefficient of $x_2^9$ in the numerator of $\det H_5$, will be negative
for $x_6$ sufficiently large.  Hence, if we choose $x_1:=1$ (or any positive value) and $x_6$ and $x_2$ sufficiently large, then the numerator of $\det H_5$ will be negative.
\end{proof}

In the remainder of this section, 
we focus on the question of whether the surface $\mathcal H$ consists of (at least generically) Hopf bifurcations.  If so, this would imply that whenever a steady state of the mixed-mechanism network switches from stable to unstable, we expect it to undergo a Hopf bifurcation leading to oscillations.  We begin our analyses of Hopf bifurcations by 
giving a criterion  for such bifurcations.

\begin{proposition} \label{prop:hopf}
Consider the dynamical system~\eqref{eq:OEs-mixed} arising from the mixed-mechanism network and any positive rate constants with $k_2 = k_6 = k_9$ and 
  $k_{10} k_{3} k_{4}-(k_{3}+k_{4}) (k_{3}+k_{7}) (k_{4}+k_{7}) > 0.$
Then there exists $(x_1^*, x_2^*, x_6^*) \in \mathbb{R}^3_{>0}$ such that  
$\det H_5(x_1^*, x_2^*, x_6^*)=0$ (in other words, $ \phi \circ \chi (x_1^*, x_2^*, x_6^*)$ is on $\mathcal{H}$).  
Moreover, for such a vector $(x_1^*, x_2^*, x_6^*) $,
the system undergoes a Hopf bifurcation with respect to $x_2$ at the steady state 
$\chi(x_1^*,x_2^*,x_6^*)$ if and only if the following inequality holds:
\begin{align} \label{eq:partial-nonzero}
	\frac{d( {\rm numerator} (\det H_5)|_{x_1=x_1^*,~x_6=x_6^*})}{dx_2} | _{x_2=x_2^*} 
	~\neq~ 0~.
\end{align}
\end{proposition}

\begin{proof} 
Fix positive rate constants for which $k_2 = k_6 = k_9$ and 
  $k_{10} k_{3} k_{4}-(k_{3}+k_{4}) (k_{3}+k_{7}) (k_{4}+k_{7}) > 0.$
  By the proofs of Proposition~\ref{prop:surface} and Theorem~\ref{thm:hopf}, 
  the function 
  $\det H_5: \mathbb{R}^3_{>0} \to \mathbb{R}$ takes both positive and negative values.
So, as 
$\det H_5$ is continuous, 
$\det H_5(x_1^*, x_2^*, x_6^*)=0$
for some 
$(x_1^*, x_2^*, x_6^*) \in \mathbb{R}^3_{>0}$ (by the intermediate-value theorem).

Assume  $\det H_5(x_1^*, x_2^*, x_6^*)=0$.  
To see  whether the steady state $\chi(x_1^*,x_2^*,x_6^*)$ is a Hopf bifurcation with respect to the parameter $\mu=x_2$, where the curve of steady states is 
$x(\mu)= \chi(x_1^*,\mu,x_6^*)$ and $\mu_0=x^*_2$, 
we use Proposition~\ref{prop:yang} (Yang's criterion).  
Parts (i) and (ii) of that criterion hold for {\em any} steady state $\chi(x_1^*,x_2^*,x_6^*)$, because 
$b_6=b_6(x_1^*,x_2^*,x_6^*)>0$, by~\eqref{eq:b6}, and also $\det H_i= \det H_i(x_1^*,x_2^*,x_6^*)>0$ for $i=1,2,3,4$ (from the proof of Proposition~\ref{prop:surface}).
Recall from the proof of Proposition~\ref{prop:surface} that the denominator of $\det H_5$ is strictly positive and does not depend on $x_2$; thus, we can focus on the numerator of $H_5$.  
So, 
by Proposition~\ref{prop:yang}, 
$\chi(x_1^*,x_2^*,x_6^*)$ is a Hopf bifurcation with respect $x_2$
if and only if \eqref{eq:partial-nonzero} holds.
\end{proof}

\begin{remark} \label{rmk:perturb}
Given  rate constants $k_i$ as in Proposition~\ref{prop:hopf} 
for which there is a Hopf bifurcation, we can perturb slightly the rate constants involved in (\ref{eq:kc_condi}) 
(while maintaining the equality $k_2 = k_6 = k_9$) and preserve the existence of a Hopf bifurcation.  
Indeed, this assertion follows from Proposition~\ref{prop:hopf} (inequality~\eqref{eq:partial-nonzero} is maintained under small perturbations of the $x_i$'s), 
the fact that simple roots of a polynomial depend continuously -- in fact, infinitely differentiably -- 
on the coefficients \cite{lc12}, and the fact that the inequality (\ref{eq:kc_condi}) defines a (relatively) open set in the parameter space of the $k_i$'s.
\end{remark}

Under the hypotheses of Proposition~\ref{prop:hopf},
we expect that inequality~\eqref{eq:partial-nonzero} holds generically on $\mathcal H$.  
We will confirm this when the rate constants are those in Table~\ref{tab:rates} 
(Theorem~\ref{thm:main}).

The proof of Theorem~\ref{thm:main} 
makes use of discriminants, which we now review.  
Consider a degree-$n$, univariate polynomial $f=c_n x^n+c_{n-1}x^{n-1}+ \cdots + c_0$ with coefficients $c_i \in \mathbb{C}$.  A {\em multiple root} of $f$ is some $x^* \in \mathbb{C}$ for which $(x-x^*)^2$ divides $f$ or equivalently $f(x^*)=f'(x^*)=0$.  It is well-known that $f$ has a multiple root in $\mathbb{C}$ if and only if a certain multivariate polynomial in the $c_i$'s, the {\em discriminant}, vanishes~\cite{Gelfand:Kapranov:Zelevinsky}.  For instance, the discriminant of the quadratic polynomial $ax^2+bx+c$ is the familiar expression $b^2-4ac$.

\begin{theorem}[Hopf bifurcations of the mixed-mechanism network] \label{thm:main}
  Consider the dynamical system~\eqref{eq:OEs-mixed} arising from the
  mixed-mechanism network and rate constants in Table~\ref{tab:rates}.
  Let $\mathcal{H}$ denote the surface, from
  Proposition~\ref{prop:surface}, 
  that defines the border between those $\invtPoly$ whose unique
  steady state $x^*$ is locally stable and those $\invtPoly$ for which
  $x^*$ is unstable.  Then  $\mathcal{H}$ consists generically of
  compatibility classes $\invtPoly$  whose unique steady state  $x^*$
  undergoes a simple Hopf bifurcation (with $x_2$ as
  bifurcation parameter).
\end{theorem}

\begin{proof}
It is straightforward to check that the rate constants in Table~\ref{tab:rates}
satisfy the inequality~\eqref{eq:kc_condi}.  Therefore, the surface~$\mathcal H$
as in 
Proposition~\ref{prop:surface}.2(b) exists, and is defined by $\det H_5=0$, where
$H_5$ is the Hurwitz matrix (specialized to the rate constants in Table~\ref{tab:rates}) as in the proof of Proposition~\ref{prop:surface}.

To prove that $\mathcal H$ consists generically of Hopf bifurcations, we use
Proposition~\ref{prop:hopf}.  That result states that
$\chi(x_1^*,x_2^*,x_6^*)$ is a Hopf bifurcation with respect to $x_2$ if and only if  
$(x^*_1,x^*_2,x^*_6) \in \mathcal{H}' \setminus \mathcal{S}$, where 
\begin{align*}
	\mathcal{H}' ~&:=~
		V_{>0}(\det H_5) ~:=~ \left\{ (x_1,x_2,x_6) \in \mathbb{R}^3_{>0} \mid 
		\det H_5(x_1,x_2,x_6)=0 \right\}~,~ {\rm and} \\
	\mathcal{S}  ~&:=~   \left\{ (x^*_1,x^*_2,x^*_6) \in \mathcal{H}' ~~\middle|~~ 
		\frac{d(\det H_5|_{x_1=x_1^*,~x_6=x_6^*})}{dx_2} | _{x_2=x_2^*}=0
			\right\}~\subseteq \mathcal{H}' ~. 
\end{align*}
We have that $\mathcal{H}=\phi \circ \chi(\mathcal{H}')$, and 
that the following subset of $\mathcal{H}$ consists of compatibility
classes whose unique steady state undergoes a simple Hopf
bifurcation with $x_2$ as bifurcation parameter:
$\phi \circ \chi(\mathcal{H}' \setminus \mathcal{S})$. 
So, 
it suffices to show that $\dim (\mathcal{S})< \dim (\mathcal{H}')$.
Note that $\dim (\mathcal{H}') \geq 2$, so we will show that $\dim (\mathcal{S}) \leq 1$.

 To this end, note that if $(x_1^*,x_2^*,x_6^*) \in \mathcal{S}$, 
 then $x_2^*$ is a multiple root of the univariate polynomial 
 ${\rm numerator}(\det H_5)|_{x_1=x_1^*,~x_6=x_6^*}$ (this also uses the fact the denominator of 
$\det H_5$, which is $188956800000000000000 x_6^5$, does not depend on $x_2$).
 Thus, any $(x_1^*,x_2^*,x_6^*) \in \mathcal{S}$ satisfies $D(x_1^*,x_6^*)=0$, where $D$ is the discriminant of $\det H_5$ and $H_5$ is viewed as a univariate polynomial in the variable $x_2$.  
So, we have the map:
	\begin{align*}
	\mathcal{S} \quad  &\to \quad \{ (x_1,x_6) \in \mathbb{R}^2 \mid D(x_1,x_6)=0\} ~=:~ \mathcal{D} \\ 
    (x_1,x_2,x_6) \quad &\mapsto \quad (x_1,x_6) ~.
    \notag
	\end{align*}
The preimage of any point of this map has size at most 4 
(because ${\rm numerator}(\det H_5)|_{x_1=x_1^*,~x_6=x_6^*}$ has degree 9, so it has at most 4 multiple roots). 

Thus, to achieve our desired inequality (namely, $\dim (\mathcal{S}) \leq 1$), 
we need only prove the following claim:
 $ \dim (\mathcal{D}) \leq 1$
or, equivalently, the bivariate polynomial $D$ is not the zero polynomial.  It suffices to show that $D(1,1)$ is nonzero, which in turn would follow if we can show that the univariate, 
degree-9 polynomial ${\rm numerator}(\det H_5)|_{x_1=x_1^*,~x_6=x_6^*}=H_5(1,x_2,1)$ does {\em not} have a multiple root over $\mathbb{C}$.  
Indeed, using {\tt Mathematica}, we see that the numerator of $\det
H_5(1,x_2,1)$ has 9 (distinct) complex roots:
\begin{align*}
-131.425, ~ -102.999, ~ 
     -78.022, ~ -66.423,  ~ 
     -39.194, ~ 
     -3.946 \pm 0.734 i, ~ 
     -3.677, ~ 268.606~.
\end{align*}
Thus, $D$ is a nonzero polynomial, and this completes the proof.
\end{proof}

\begin{figure}[ht]

  \begin{subfigure}{0.3\textwidth}
    \includegraphics[width=\textwidth]{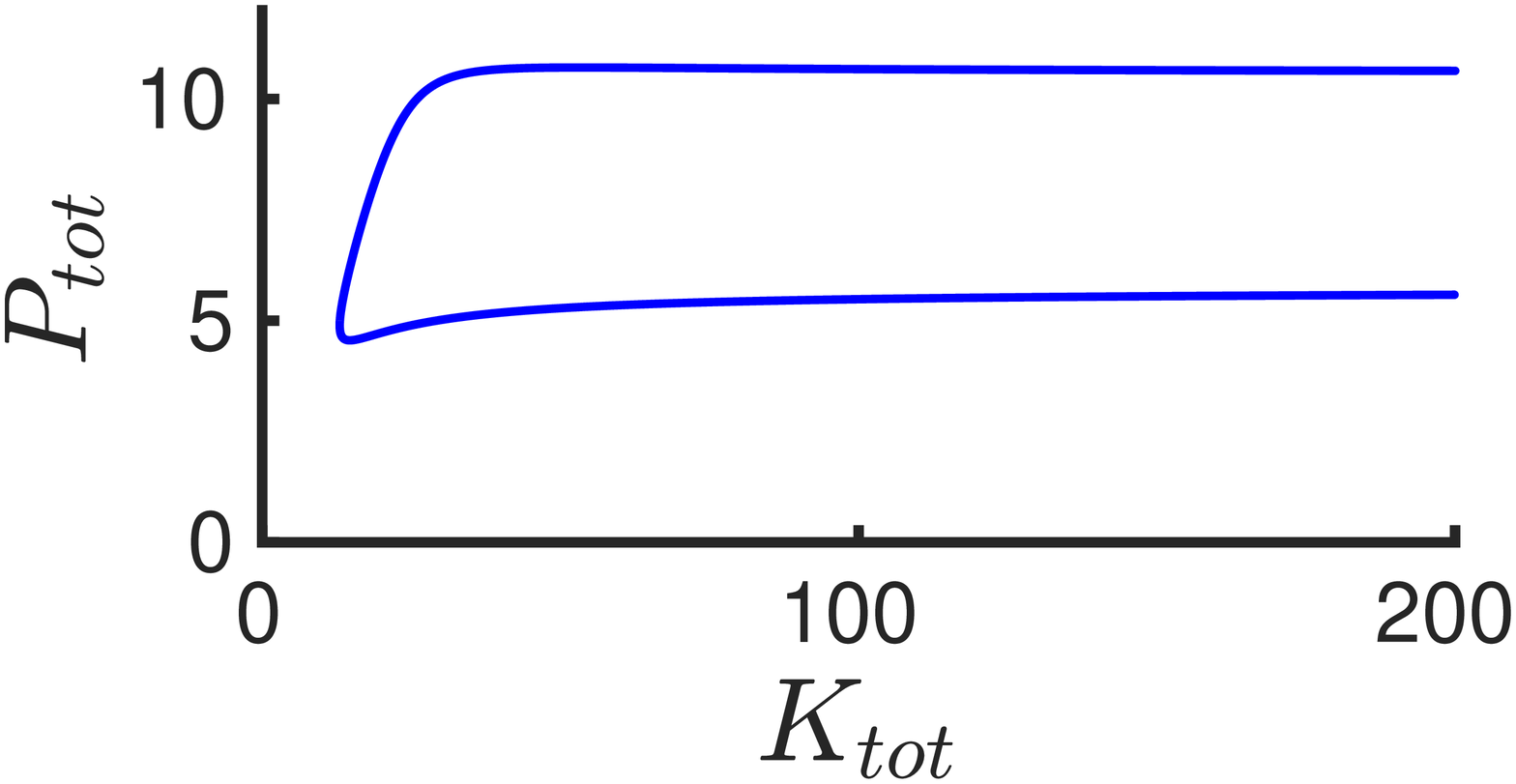}
    \subcaption{
      $\Stot =  40$.
    }
  \end{subfigure}
  \hfill
  \begin{subfigure}{0.3\textwidth}
    \includegraphics[width=\textwidth]{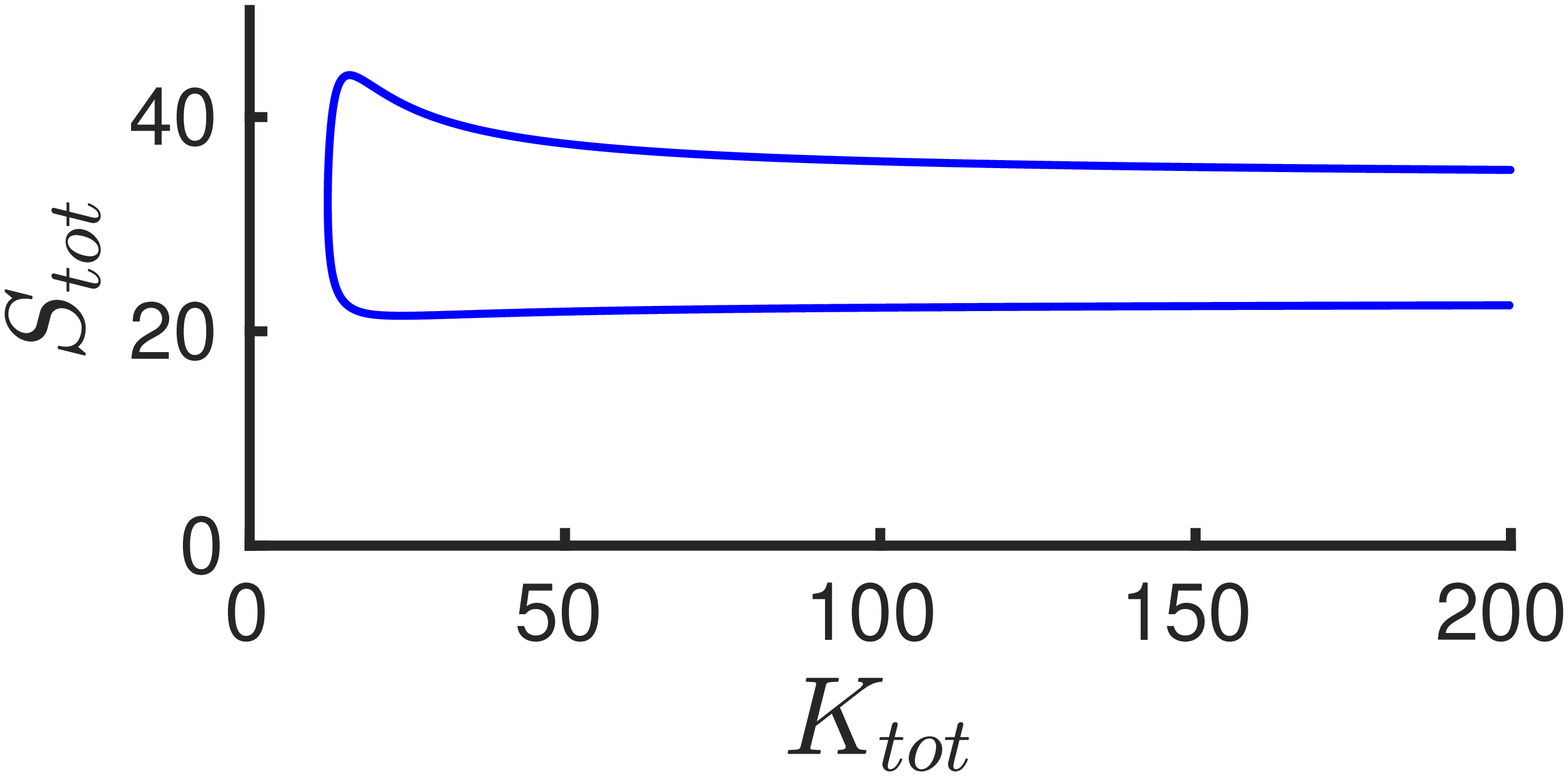}
    \subcaption{
      $\Ptot = 5$. 
    }    
  \end{subfigure}
  \hfill
  \begin{subfigure}{0.3\textwidth}
    \includegraphics[width=\textwidth]{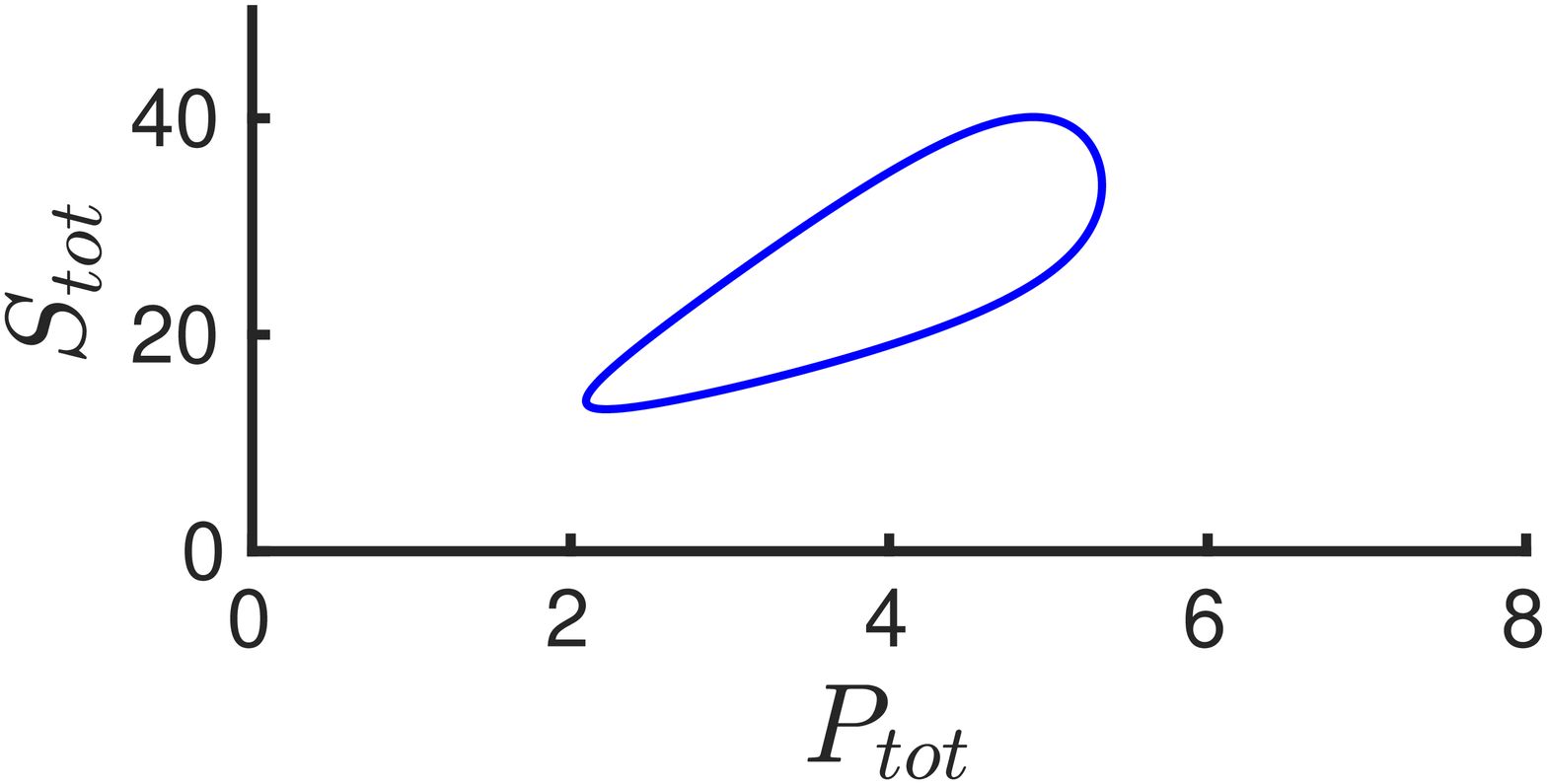}
    \subcaption{
      $\Ktot \approx
    13.0296$. 
    }    
  \end{subfigure}
  \caption{
    \label{fig:2-para-conti}
    Slices of the Hopf-bifurcation surface $\mathcal{H}$, from
    Theorem~\ref{thm:main}.  Specifically, displayed are the
    intersections of $\mathcal{H}$ with the hyperplanes defined by
    (a) $\Stot =  40$, (b) $\Ptot = 5$, and (c) $\Ktot \approx
    13.0296$.
	Each such curve was obtained numerically, using {\tt Matcont}~\cite{matcont}, by a two-parameter
	continuation of the Hopf bifurcation arising from 
	$\Ktot \approx 13.0296$, $\Ptot = 5$, and $\Stot=40$.
        Each point of the curves in (a) -- (c) corresponds to a Hopf
    bifurcation with respect to either of the two varying total
    concentrations. 
 Points ``inside'' $\mathcal{H}$ correspond to unstable steady states and thus the potential for oscillations.
  }
\end{figure}

In Figure~\ref{fig:2-para-conti}, we show some slices of the
Hopf-bifurcation surface $\mathcal{H}$ (where the rate constants are from Table~\ref{tab:rates}). 
Accordingly, this figure extends the one-dimensional Figure~\ref{fig:hopf}.

The bifurcations analyzed in Proposition~\ref{prop:hopf}
and Theorem~\ref{thm:main} 
are with respect to the bifurcation parameter $x_2$, the steady-state value of the kinase $K$.  It is natural to ask whether we also obtain a bifurcation with respect to a more biologically meaningful parameter, such as a rate constant or a total amount.  We now explain how to perform such an analysis.

To use a total amount (here we use $P_{\rm tot}$) as a bifurcation parameter 
(perturbing this parameter corresponds to perturbing the compatibility class),  consider the following maps:
\begin{align*}
\{\left( \Ktot,~\Ptot,~\Stot \right) \}	= \mathbb{R}^{3}_{>0} ~		&\overset{\phi \circ \chi}\longleftarrow \quad  \mathbb{R}_{>0}^{3} \quad 
	 \overset{h_5:= \det H_5}\longrightarrow ~
	\mathbb{R}_{>0}
	\end{align*}
Recall that $ (\phi \circ \chi):	\mathbb{R}^{3}_{>0} \to \mathbb{R}^3_{>0}$ is a bijection.
Let $\mathfrak g :=h_5 \circ (\phi \circ \chi)^{-1}: \mathbb{R}^{3}_{>0} \to \mathbb{R}$. 
Also, let $p:=(\phi \circ \chi)_{2}=x_6+\frac{1009}{1800}x_1 x_2$ denote the second coordinate function of $\phi \circ \chi$ from~\eqref{eq:param-c-class} (here we assume the rate constants from Table~\ref{tab:rates}).  We are interested in checking whether $\frac{\partial \mathfrak g}{\partial \Ptot}$ is (generically) nonzero whenever $\mathfrak g =0$.  Accordingly, we use the chain rule:
\begin{align} \label{eq:sum-partials} \notag
	\frac{\partial \mathfrak g}{\partial \Ptot}
	~&=~
	 \frac{1}{\partial p / \partial x_1}\frac{\partial h_5}{\partial x_1} ~+~ 
	 \frac{1}{\partial p / \partial x_2}\frac{\partial h_5}{\partial x_2} ~+~
	 \frac{1}{\partial p / \partial x_6}\frac{\partial h_5}{\partial x_6}
		\\
	~&=~		\frac{1800}{1009 x_2} \frac{\partial h_5}{\partial x_1} ~+~ 
	\frac{1800}{1009 x_1}\frac{\partial h_5}{\partial x_2} ~+~
	\frac{\partial h_5}{\partial x_6}~.
\end{align}
For specific values of $x_1, x_2, x_6$, it is straightforward to 
check whether the sum~\eqref{eq:sum-partials} is nonzero.  More generally, 
we expect this sum to be generically nonzero; that is, we expect that the surface $\mathcal H$ consists generically of Hopf bifurcations with respect to the total-amount $\Ptot$.


\section{Generating rate constants admitting oscillations}
\label{sec:find-param}

The proof of 
Theorem~\ref{thm:hopf}
yields a recipe for generating rate constants 
for the mixed-mechanism network 
at which we expect oscillations arising from a Hopf bifurcation.
Specifically, 
we choose rate constants $k_i$ for which 
the equalities $k_2 = k_6 = k_9$ hold, the 
 inequality~\eqref{eq:kc_condi} holds, and
  $\alpha_0<0$ (as in~\eqref{eq:alpha-0}), 
  and then pick $x_2$ and $x_6$ 
  large enough so that $\det H_5$ is negative but close to 0.
We summarize these choices in the following procedure.

\begin{procedure}[Generating rate constants likely to admit oscillations] \label{proc:rates}
~

 \noindent
{\bf Input:} 
The following functions\footnote{The functions are provided as a text file in the Supporting Information. See Appendix~\ref{sec:A}.}: 
\begin{enumerate}[(i)]
	\item 
	$\alpha_0$ as in~\eqref{eq:alpha-0}, 
	\item 
	 the numerator of $\det H_5$,
	\item 	 
	$q:= \alpha_0 x_6^3 + \alpha_1 x_6^2 + \alpha_2 x_6 +
    \alpha_3$ as in~\eqref{eq:h5}, and 
	\item 
    $\phi \circ \chi$ given in Proposition~\ref{prop:param-comp-classes}.
\end{enumerate}
\noindent
{\bf Output:} Rate constants and total amounts for which $\det H_5$ is negative and close to 0.

\noindent
{\bf Steps:} 
\begin{enumerate}
	\item Choose positive values for  $k_b:=k_2=k_6=k_9$, $x_1$, $k_1$,
  $k_3$, $k_4$, $k_7$, and $k_8$. 
	\item Choose a positive value for $k_{10}$ for which 
	$ k_{10} > \frac{(k_{3}+k_{4}) (k_{3}+k_{7}) (k_{4}+k_{7})}{k_{3}
      k_{4}}. $
	\item Choose the remaining rate constant $k_{5}$ such that $\alpha_0<0$.
	\item Choose $x_6$ so that $q<0$.
	\item Choose $x_2$ so that the numerator of $\det H_5$ is negative but close to 0.
	\item 
	 Return the $k_i$'s and $(\Ktot,~\Ptot,~\Stot):=\phi \circ \chi (  x_1,x_2,x_6)$, where 
$\phi \circ \chi$ is evaluated at the $k_i$'s (and $x_1,x_2,x_6$) chosen in the previous steps.
\end{enumerate}
\end{procedure}

\begin{remark} \label{rmk:use-matcont}
Using the output of Procedure~\ref{proc:rates},
one can attempt to exhibit and analyze oscillations 
or Hopf bifurcations using software, e.g., {\tt Matcont}~\cite{matcont}.  See Figure~\ref{fig:hopf-procedure}.
\end{remark}

\begin{example} \label{ex:procedure}
We follow Procedure~\ref{proc:rates} as follows (to verify our computations see the file
\texttt{mixed\_generate\_rc.nb}):

\underline{Step 1}. We pick $k_b=0.143738$, $k_1=0.575284$, $k_3=3.89096$,
  $k_4=5.05386$, $k_7=9.25029$, $k_8=0.621813$, and $x_1=5.82148$.

\underline{Step 2}.  The inequality for this step evaluates to 
$    k_{10} > 85.5048$, so we choose $k_{10}=90$.

\underline{Step 3}. Evaluating $\alpha_0$ at the chosen $k_i$'s, we obtain the following
  inequality: 
  \begin{displaymath}
    -8.896\times 10^{17} k_{5}^3+1.49735\times 10^{20}
    k_{5}^2+4.79701\times 10^{20} k_{5}+2.42695\times 10^{20} < 0~,
  \end{displaymath}
which we find, using {\tt Mathematica}, is feasible for
$    k_5 > 171.471$.  So, we pick $k_5=172$.

\underline{Step 4}. 
 By evaluating $q$ at the values chosen above, we obtain the following inequality:
  \begin{displaymath}
    -1.41683\times 10^{22} x_{6}^3-3.5508\times 10^{25}
    x_{6}^2-1.80374\times 10^{25} x_{6}+2.15078\times 10^{24}
    < 0~.
  \end{displaymath}
 This inequality holds when 
$    x_6 > 0.0996797$, so we choose $x_6=0.1$.

\underline{Step 5}. 
By evaluating the numerator of $\det H_5$, 
  we obtain the following inequality:
  \begin{displaymath}
    \begin{split}
      &-5.42893\times 10^{25} x_{2}^9-4.20944\times 10^{29}
      x_{2}^8-5.05393\times 10^{31} x_{2}^7-6.67609\times 10^{32}
      x_{2}^6 \\
      &\quad 
      +4.66164\times 10^{33} x_{2}^5+3.97617\times 10^{34}
      x_{2}^4+1.01289\times 10^{35} x_{2}^3+1.19894\times 10^{35}
      x_{2}^2 \\
      &\quad 
      +6.7831\times 10^{34} x_{2}+1.4718\times 10^{34} < 0~.
    \end{split}
  \end{displaymath}
This inequality is feasible, as computed in {\tt Mathematica},  for
$    x_2 > 9.0382$; we pick $x_2=10$.
  
\underline{Step 6}. 
We have determined the following rate constants:
\begin{table}[!h]
  \centering
  \begin{tabular}{|c|c|c|c|c|c|c|c|c|c|} \hline
    $k_{1}$ &  $k_{2}$ &  $k_{3}$ & $k_{4}$ &  $k_{5}$ &  $k_{6}$ & $k_{7}$ & $k_{8}$ &  $k_{9}$ &  $k_{10}$
     \\ \hline 
    0.575284 & 0.143738 & 3.89096 & 5.05386 & 172 & 0.143738 & 9.25029 & 0.621813 & 0.143738 & 90
    \\ \hline
  \end{tabular}
\end{table}

We obtain the following steady state, using~\eqref{eq:param}:
\begin{align} \label{eq:steady-state-procedure}
(x_1,x_2,\dots, x_9)
	~&=~\chi(x_1,x_2,x_6) \\ \notag
	~&=~   
	( 5.82148 ,~ 10 ,~ 8.30052 ,~ 6.39056 ,~ 1.90691 ,~ 0.1 ,~ 3.49146 ,~ 520.229 ,~ 0.358855)~.
\end{align}
Using this steady state, we obtain the total amounts, using~\eqref{eq:phi}:
\begin{align} \label{eq:totals-procedure}
(\Ktot,~\Ptot,~\Stot) ~&=~ \phi(x_1,x_2,\dots,x_9)
	 ~=~
	 (24.6911,~ 3.95031,~546.499)~.
\end{align}
The resulting bifurcation analysis is shown in Figure~\ref{fig:hopf-procedure}.
\end{example}

\begin{figure}[h] 

  \begin{subfigure}{0.3\textwidth}
    \includegraphics[width=\textwidth]{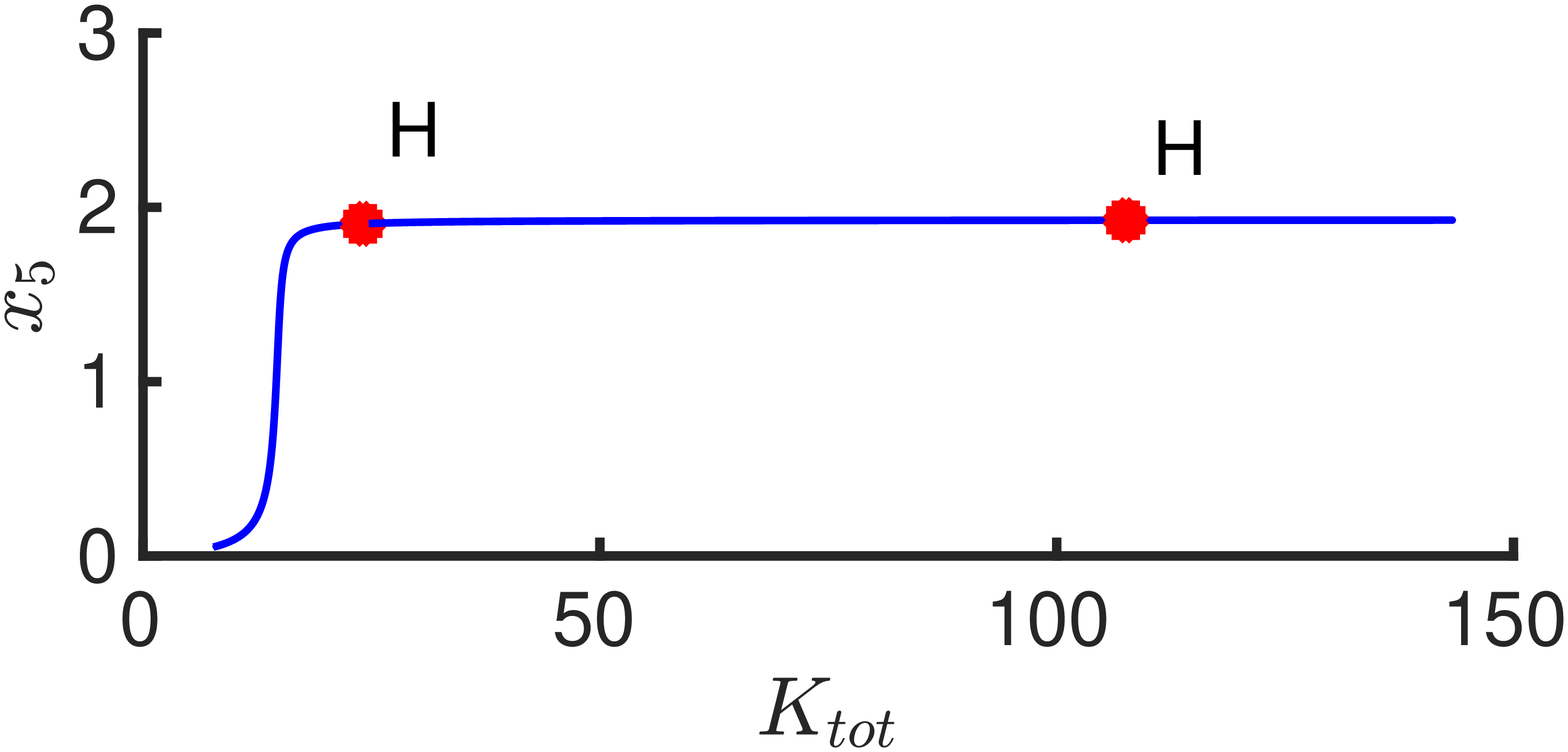}
    \subcaption{
      Bif.~parameter $\Ktot$.
    }
  \end{subfigure}
  \hfill
  \begin{subfigure}{0.3\textwidth}
    \includegraphics[width=\textwidth]{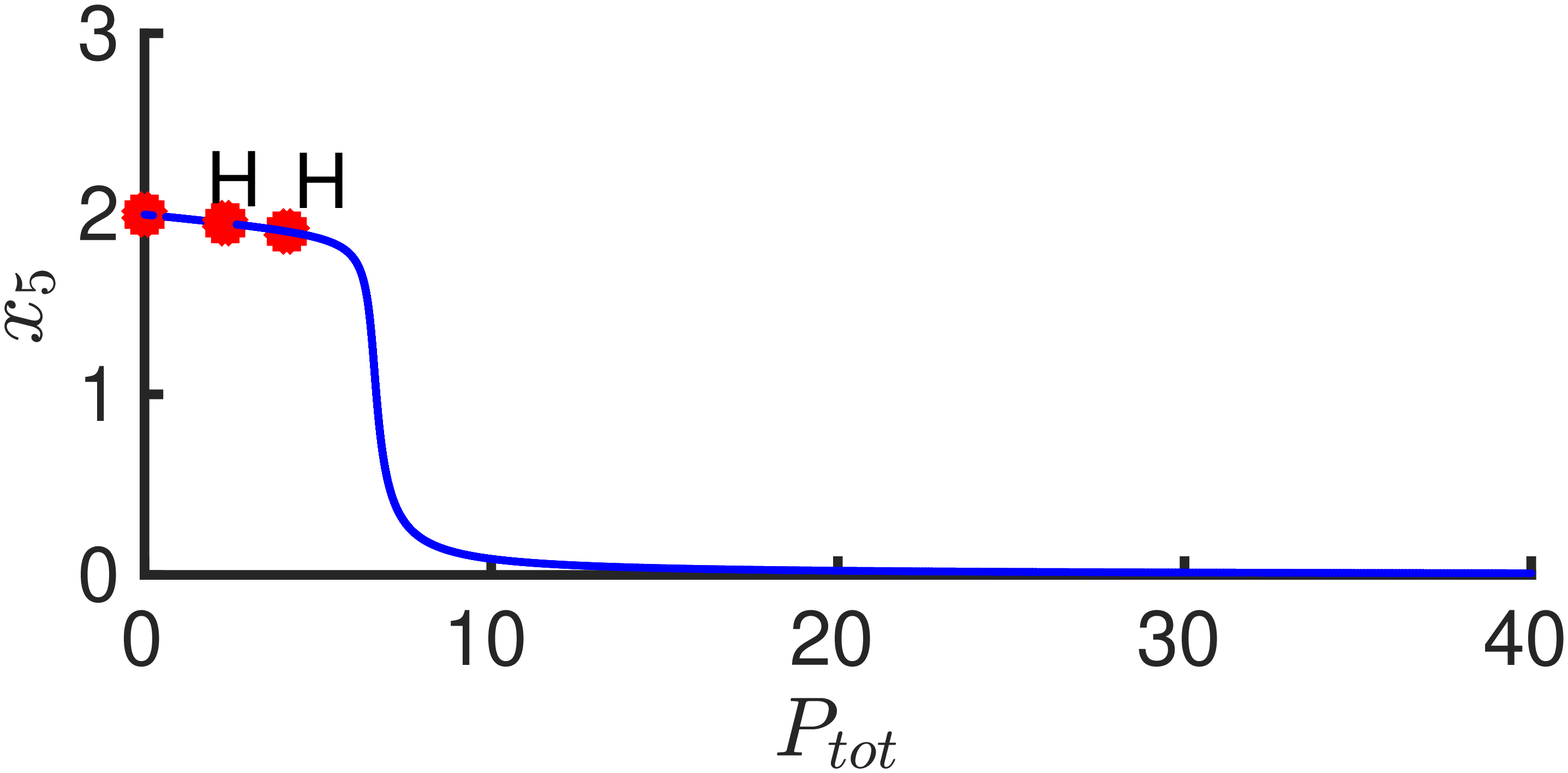}
    \subcaption{
      Bif.~parameter $\Ptot$.
    }    
  \end{subfigure}
  \hfill
  \begin{subfigure}{0.3\textwidth}
    \includegraphics[width=0.9\textwidth]{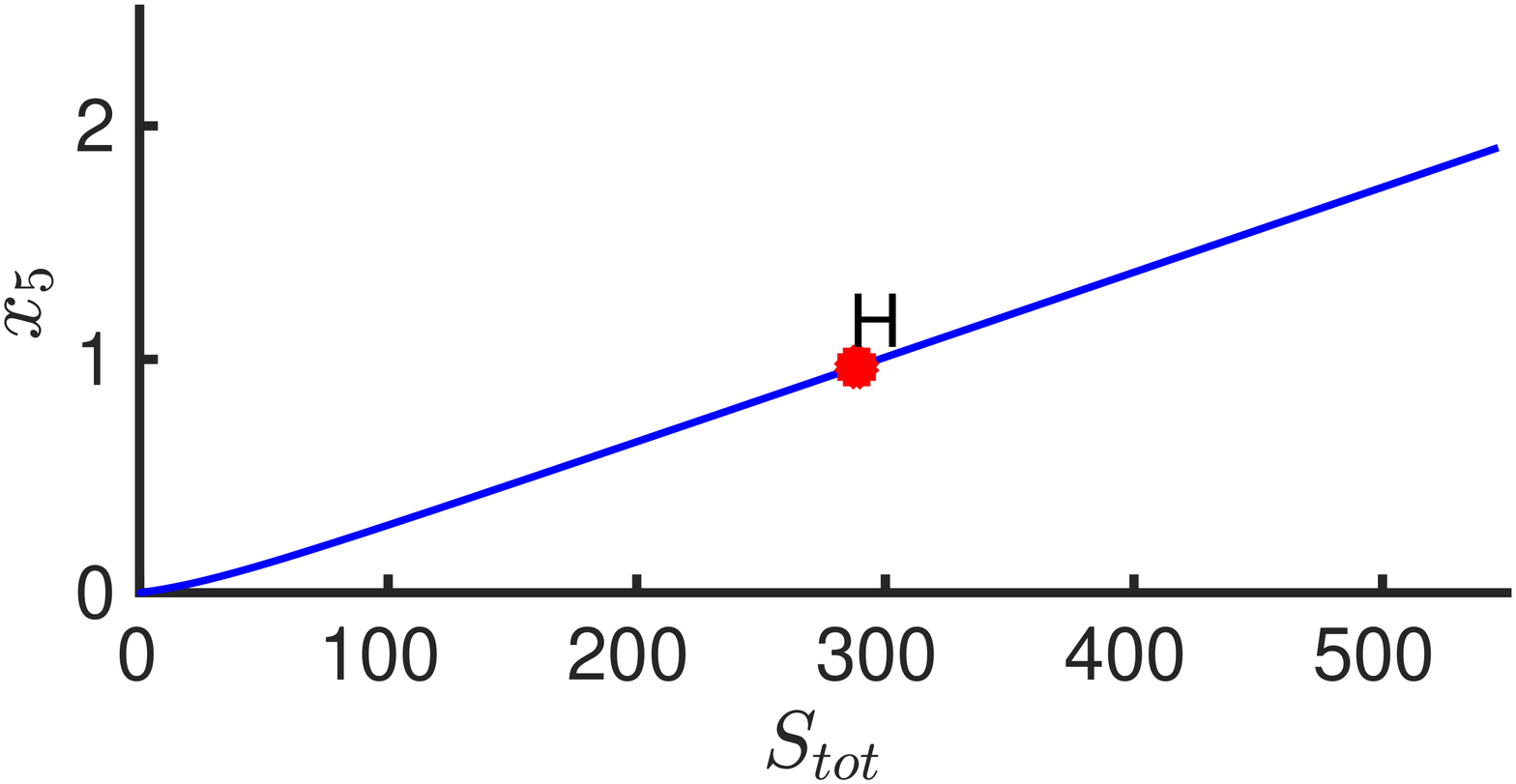}
    \subcaption{
      Bif.~parameter $\Stot$. 
    }    
  \end{subfigure}
  \caption{
    \label{fig:hopf-procedure}
    Numerical continuation of the 
    steady state~\eqref{eq:steady-state-procedure},
    when total amounts are as in~\eqref{eq:totals-procedure}:
    (a) A (supercritical) Hopf
    bifurcations are at $\Ktot \approx 24.0623$ and $107.5635$.
    (b) (Supercritical)
    Hopf
    bifurcations are at $\Ptot \approx 4.1022$ and $\Ptot \approx
    2.3275$. {\tt Matcont} reported a branch point, the leftmost red circle, at $P_{tot} \approx
    -8.5427 \times 10^{-13}$, i.e., for $\Ptot  \approx 0$  and thus outside
    the domain of interest. 
    (c) A (supercritical) Hopf bifurcation is at $\Stot
    \approx 288.4384$.
}
\end{figure}

\section{Dynamics: simulations and conjectures} \label{sec:simulations}
Are oscillations the norm when the mixed-mechanism system has an unstable steady state?  We conjecture that this is the case.
\begin{conjecture} \label{conj:osc}
Consider the mixed-mechanism network, and any choice of rate constants and total amounts.  
If the unique steady state in $\invtPoly$ is unstable, then $\invtPoly$ contains a  stable periodic orbit.
\end{conjecture}

\begin{figure}[ht]

  \begin{subfigure}{0.3\textwidth}
    \includegraphics[width=0.9\textwidth]{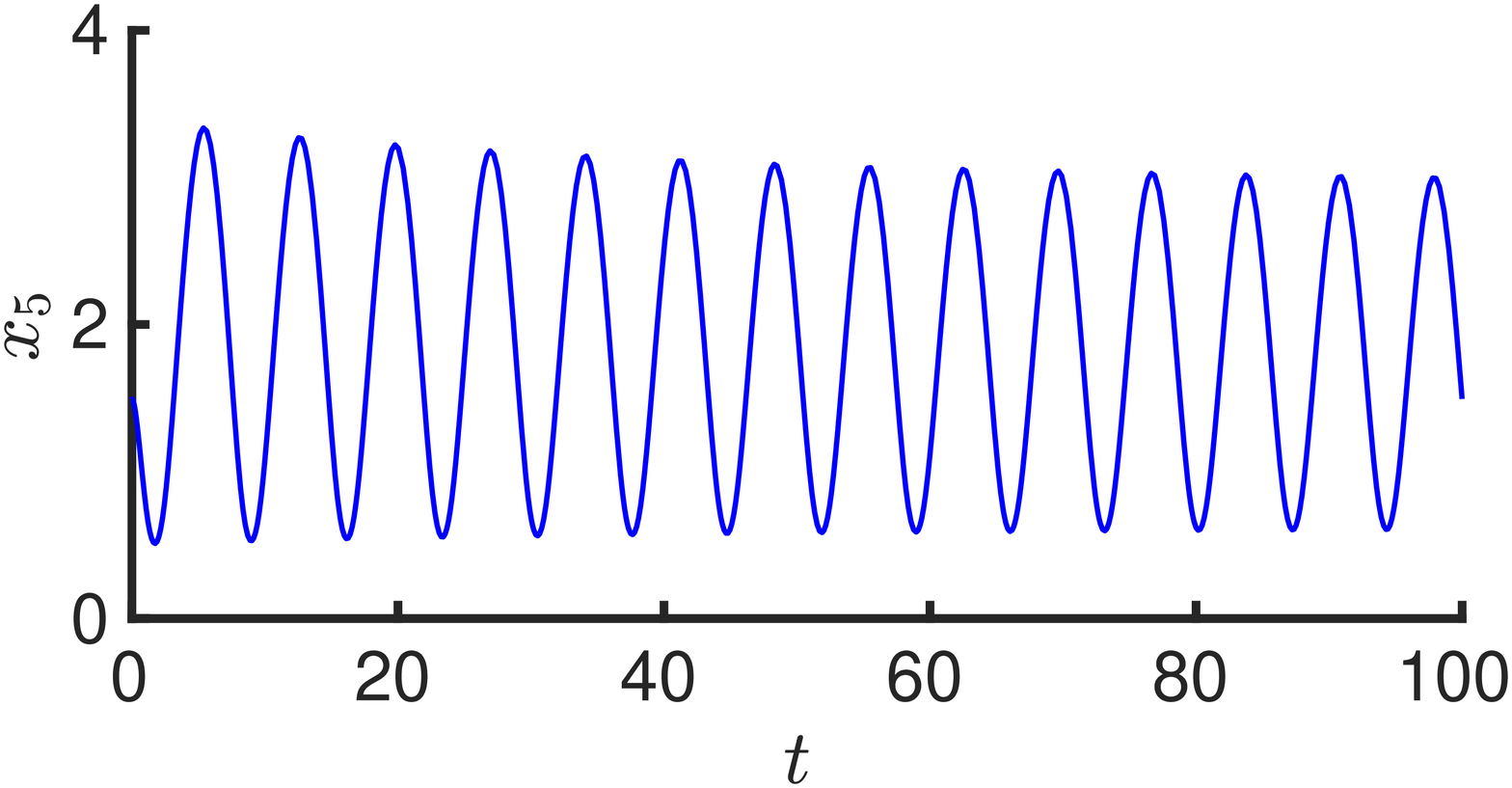}
    \subcaption{
      $x_5$ vs.~$t$.
    }
  \end{subfigure}
  \begin{subfigure}{0.3\textwidth}
    \includegraphics[width=0.9\textwidth]{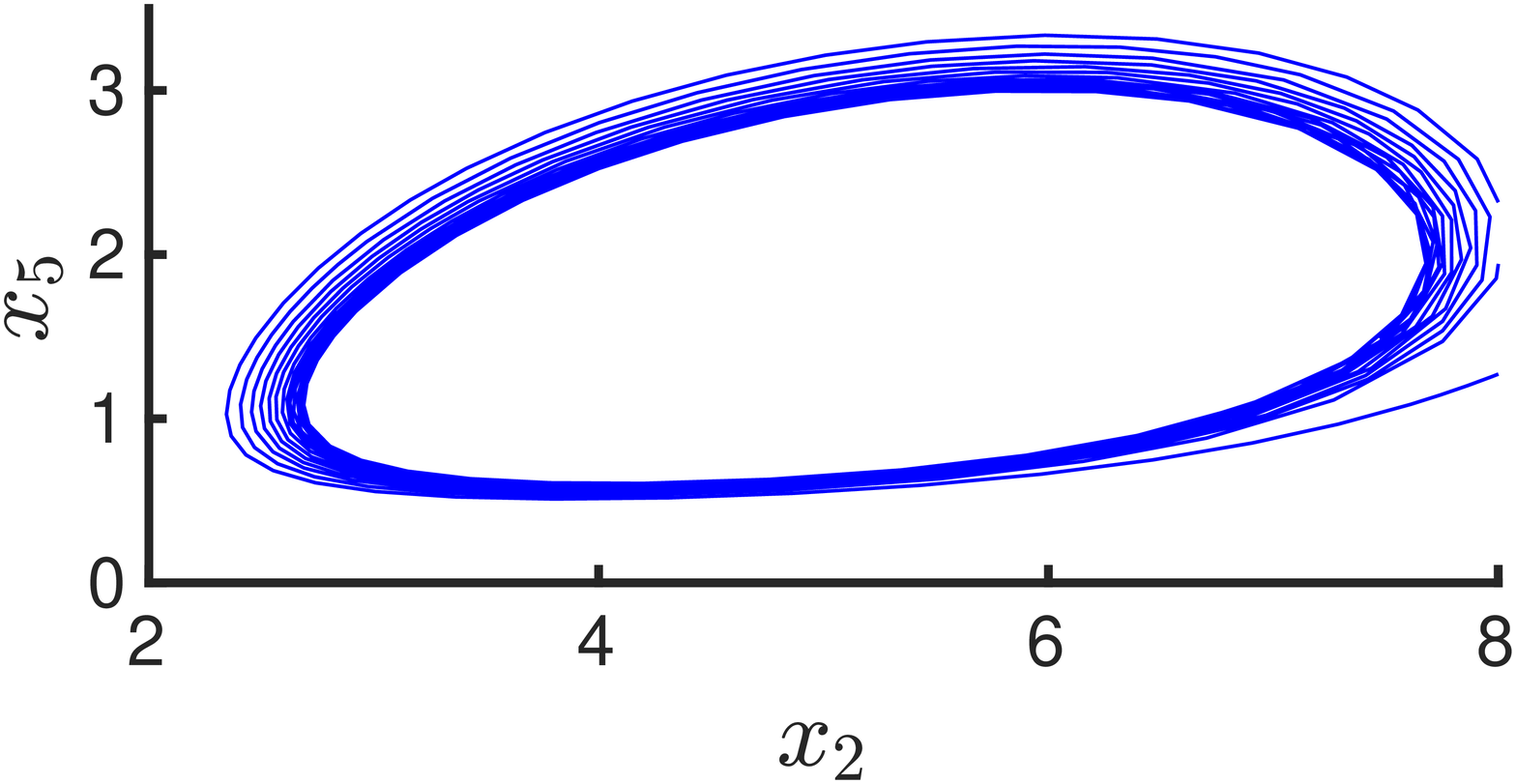}
    \subcaption{
      $x_5$ vs.~$x_2$.
    }    
  \end{subfigure}
  \begin{subfigure}{0.3\textwidth}
    \includegraphics[width=0.9\textwidth]{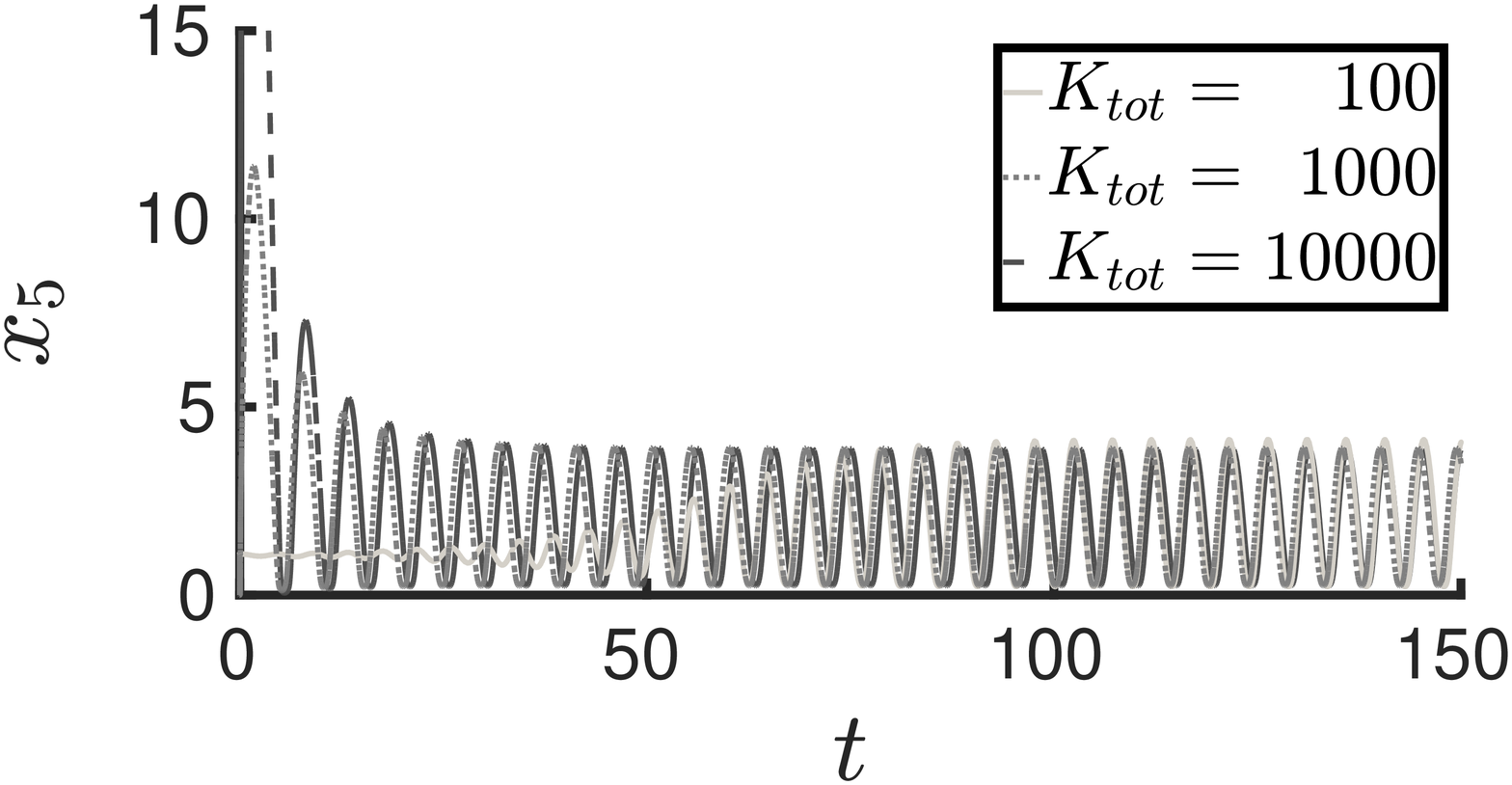}
    \subcaption{
      Increasing $K_{tot}$.
    }    
  \end{subfigure}
  \caption{
    \label{fig:oscis}
    Numerical verification of oscillations in the mixed-mechanism system with rate constants as in Table~\ref{tab:rates}.  
    For (a) and (b), we used $(\Ktot,\Ptot,\Stot)=(14, 5, 40)$ and 
    initial values as in~\eqref{eq:steady-ste-for-17-5-40}. 
    Here the solution converges to a periodic orbit.
    For (c), we used $(\Ptot,\Stot)=(8, 40)$ and 
    three values for $\Ktot$ (namely, 100, 1000, and 10000),
    and again  initial values as in~\eqref{eq:steady-ste-for-17-5-40}, except
    that $x_5=1.1$.
    Again the solutions seem to converge to a periodic orbit, and
    moreover this periodic orbit appears  not to depend on the value
    of $\Ktot$.  See Conjecture \ref{conj:k-tot}.
   }
\end{figure}

Some simulations are shown in Figure~\ref{fig:oscis}.  In (A) and (B) of that figure, we see 
solutions converging to a period orbit; this system arises from total-amounts similar to those that 
Suwanmajo and Krishnan
found to support oscillations.  
In contrast, in Figure~\ref{fig:oscis}(C),
 we see oscillations,  when $(\Ptot,\Stot)=(8,40)$,
 for  
three large values for $K_{\rm tot}$: 100, 1000, and 10000.  
Oscillations persist across these values,
which yields a much larger range for $K_{\rm tot}$ than Suwanmajo and Krishnan's results would suggest.

Moreover, the value of $K_{\rm tot}$ appears {\em not} to affect the resulting periodic orbit (when projected to $x_5$, the concentration of the doubly phosphorylated substrate $S_2$). 
Could this be a biological design mechanism for robust timekeeping (for instance, in circadian clocks)?  Mathematically, we conjecture that  oscillations indeed persist for arbitrarily large $K_{\rm tot}$;  and, that  the periodic orbit in $x_5$ indeed does  not depend on $K_{\rm tot}$. 
\begin{conjecture} \label{conj:k-tot} ~
\begin{enumerate}
	\item Consider the mixed-mechanism network with rate constants as in Table~\ref{tab:rates}.  Then there exist values of $\Ptot$ and $\Stot$ such that for $\Ktot$ arbitrarily large, the unique steady state in $\invtPoly$ is unstable.
	\item For such values of $\Ptot$ and $\Stot$ and for sufficiently large $\Ktot$, the compatibility class $\invtPoly$ contains a periodic orbit such that this orbit in $x_5$ (the concentration of $S_2$) does not depend on the value of $\Ktot$.
\end{enumerate}
\end{conjecture}

One way to tackle Conjecture~\ref{conj:k-tot} is to
analyze the robustness of the period and the amplitude with respect to $K_{\rm tot}$ using the theory developed 
in \cite{br74,bi04,imr17}.

Finally, we consider the dynamics in compatibility classes that contain a locally stable steady state.  Our simulations  suggest that such a steady state is in fact globally stable.
Accordingly, we pose the question, {\em Consider the mixed-mechanism network, and any choice of rate constants and total amounts.  
If the unique steady state $x^*$ in $\invtPoly$ is locally stable, does it always follow that $x^*$ is globally stable?}  In the Michaelis-Menten limit, this is true~\cite{Rao-2}.

\section{Discussion} \label{sec:discussion}
We return to the question, {\em How do oscillations emerge in phosphorylation networks?}  Concretely, we would like (1) easy-to-check criteria for exactly which phosphorylation networks admits oscillations or Hopf bifurcations, and
(2) for those networks that admit oscillations, a better understanding of the ``geography of parameter space'', that is, a characterization of which rate constants and initial conditions yield oscillations.  Both of these problems are still unresolved, and the second problem in particular is very difficult.

Nevertheless, here we made progress on characterizing some of the geography of parameter space for the mixed-mechanism phosphorylation network.  Indeed, we found that a single surface defines the boundary between stable and unstable steady states, and this surface consists generically of Hopf bifurcations.  Hence, when a steady state switches from stable to unstable, then we expect it to undergo a Hopf bifurcation leading to oscillations.  Additionally, we gave a procedure for generating many parameter values leading to oscillations.  

  We now discuss the significance of our work.  At a glance, it might
  seem that our results are specific to
  network~\eqref{eq:mixed-network} and rate constants related to those in
  Table~\ref{tab:rates}.  However, the approach is general: for other
  rate constants (e.g., estimated from data) or other networks (e.g., a
  version of the ERK network from~\cite{long-term} also
  has oscillations and a unique steady state), one could apply the
  same techniques. Therefore, the potential impact is broad.

Going forward, we hope that the novel techniques we used -- specifically, using a steady-state parametrization together with a Hopf-bifurcation criterion -- will contribute to solving other problems.  For instance, 
we expect that such tools could help solve an important open problem in this area~\cite{perspective}, namely, the question of whether oscillations or Hopf bifurcations arise from the fully distributive phosphorylation network.

\subsection*{Acknowledgements}
AS was partially supported by the NSF (DMS-1312473/1513364 and DMS-1752672)
and the Simons Foundation (\#521874). AS thanks Alan Rendall and Jonathan Tyler for helpful discussions.
CC was partially supported by the Deutsche Forschungsgemeinschaft DFG
(DFG-284057449). The authors two referees for their helpful suggestions.  

\appendix

\section{Files in the Supporting Information} \label{sec:A}
The following files 
  can be found 
  at \url{http://www.math.tamu.edu/~annejls/mixed.html}:

\underline{Text files}:
\begin{itemize}
\item \texttt{mixed\_H5N\_kb.txt} \ldots contains \texttt{H5N}, the
  numerator of $\det H_5$ under the assumption $k_2=k_6=k_9=k_b$
\item \texttt{mixed\_W.txt} \ldots contains a matrix \texttt{W} that
  defines (\ref{eqn:conservation})
\item \texttt{mixed\_xt.txt} \ldots contains {\tt xt}, the
  parameterization (\ref{eq:param})
\item \texttt{mixed\_Jx.txt} \ldots contains {\tt Jx}, the Jacobian
  evaluated at the parameterization (\ref{eq:param})
\end{itemize}

\underline{\texttt{Mathematica} Notebooks:}
\begin{itemize}
\item \texttt{mixed\_analysis\_H5N\_x1\_LT.nb}:\\
  \underline{Functionality:} This file can be used to obtain
  $\text{numerator}(\det H_5)$ as in (\ref{eq:numerator-x1}), in
  particular to examine the coefficients $\alpha_{01}$, $\alpha_{10}$,
  \ldots \\ 
  \underline{Input:} the file \texttt{mixed\_H5N\_kb.txt}
\item \texttt{mixed\_analysis\_H5N\_x2\_LT.nb}: \\
\underline{Functionality:} This file can be used to obtain
  $\text{numerator}(\det H_5)$ as in (\ref{eq:h5}), in
  particular to examine the coefficients $\alpha_{0}$, \ldots,
  $\alpha_{3}$ and $\beta_0$, \ldots, $\beta_3$. \\ 
  \underline{Input:} the file \texttt{mixed\_H5N\_kb.txt}
\item \texttt{mixed\_coeffs\_charpoly.nb}: \\
  \underline{Functionality:} This file can be used to obtain
  the characteristic polynomial of the Jacobian of the system
  (\ref{eq:OEs-mixed}). It contains the {\tt Mathematica} commands to
  establish $b_i>0$. \\ 
  \underline{Input:} the file \texttt{mixed\_Jx.txt}
\item \texttt{mixed\_Hi.nb}: \\
  \underline{Functionality:} This file can be used to obtain
  the determinants of the Hurwitz matrices $H_2$, \ldots, $H_5$. It
  contains the {\tt Mathematica} commands to establish $\det H_i >0$, for
  $i=2$, $3$, $4$ and that $\det H_5$ is of mixed sign. \\
  \underline{Input:} the file \texttt{mixed\_Jx.txt}
\item \texttt{mixed\_generate\_rc.nb}: \\
  \underline{Functionality:} This file contains a realization of
  Procedure~\ref{proc:rates}. \\
  \underline{Input:} the files \texttt{mixed\_H5N\_kb.txt},
  \texttt{mixed\_W.txt}, \texttt{mixed\_xt.txt},
  \texttt{mixed\_Jx.txt}.
\end{itemize}

\bibliographystyle{amsplain}

\bibliography{dual-site}

\providecommand{\bysame}{\leavevmode\hbox to3em{\hrulefill}\thinspace}
\providecommand{\MR}{\relax\ifhmode\unskip\space\fi MR }
\providecommand{\MRhref}[2]{%
  \href{http://www.ams.org/mathscinet-getitem?mr=#1}{#2}
}
\providecommand{\href}[2]{#2}
\begin{thebibliography}{10}

\bibitem{Aoki}
Kazuhiro Aoki, Masashi Yamada, Katsuyuki Kunida, Shuhei Yasuda, and Michiyuki
  Matsuda, \emph{Processive phosphorylation of {ERK} {MAP} kinase in mammalian
  cells}, P.\ Natl.\ Acad.\ Sci.\ USA \textbf{108} (2011), no.~31,
  12675--12680.

\bibitem{adp18}
Peter Atkins, Julio De~Paula, and James Keeler, \emph{Atkins' physical
  chemistry}, Oxford University Press, 2018.

\bibitem{br74}
EG~Bure and Ye~N Rozenvasser, \emph{On investigations of autooscillating system
  sensitivity}, Avtomat. i Telemekh (1974), no.~7, 9--17.

\bibitem{CFMW}
Carsten Conradi, Elisenda Feliu, Maya Mincheva, and Carsten Wiuf,
  \emph{Identifying parameter regions for multistationarity}, PLoS Comput.\
  Biol. \textbf{13} (2017), no.~10, e1005751.

\bibitem{a6maya}
Carsten Conradi and Maya Mincheva, \emph{Catalytic constants enable the
  emergence of bistability in dual phosphorylation}, J. R. Soc. Interface
  \textbf{11} (2014), no.~95.

\bibitem{ConradiShiu}
Carsten Conradi and Anne Shiu, \emph{A global convergence result for processive
  multisite phosphorylation systems}, B.\ Math.\ Biol. \textbf{77} (2015),
  no.~1, 126--155. \MR{3303108}

\bibitem{perspective}
\bysame, \emph{Dynamics of post-translational modification systems: recent
  progress and future challenges}, Biophys.\ J. \textbf{114} (2018), no.~3,
  507--515.

\bibitem{matcont}
Annick Dhooge, Willy Govaerts, and Yuri~A. Kuznetsov, \emph{{MATCONT}: A
  {MATLAB} package for numerical bifurcation analysis of {ODEs}}, ACM Trans.\
  Math.\ Softw. \textbf{29} (2003), no.~2, 141--164.

\bibitem{Domijan2009}
Mirela Domijan and Markus Kirkilionis, \emph{Bistability and oscillations in
  chemical reaction networks}, J.\ Math.\ Biol. \textbf{59} (2009), no.~4,
  467--501.

\bibitem{EithunShiu}
Mitchell Eithun and Anne Shiu, \emph{An all-encompassing global convergence
  result for processive multisite phosphorylation systems}, Math.\ Biosci.
  \textbf{291} (2017), 1--9.

\bibitem{Errami}
Hassan Errami, Markus Eiswirth, Dima Grigoriev, Werner~M. Seiler, Thomas Sturm,
  and Andreas Weber, \emph{Detection of {Hopf} bifurcations in chemical
  reaction networks using convex coordinates}, J. Comput. Phys. \textbf{291}
  (2015), 279--302.

\bibitem{fha14}
James~E Ferrell and Sang~Hoon Ha, \emph{Ultrasensitivity part {II}: multisite
  phosphorylation, stoichiometric inhibitors, and positive feedback}, Trends
  Biochem. Sci. \textbf{39} (2014), no.~11, 556--569.

\bibitem{fg59}
Feliks~R. Gantmacher, \emph{Matrix theory}, Chelsea, New York \textbf{21}
  (1959).

\bibitem{gatermann-hopf}
Karin Gatermann, Markus Eiswirth, and Anke Sensse, \emph{Toric ideals and graph
  theory to analyze {Hopf} bifurcations in mass action systems}, J.\ Symbolic
  Comput. \textbf{40} (2005), no.~6, 1361--1382.

\bibitem{Gelfand:Kapranov:Zelevinsky}
I.M. Gelfand, M.M. Kapranov, and A.V. Zelevinsky, \emph{Discriminants,
  resultants and multidimensional determinants}, Birkh{\"a}user, 1994.

\bibitem{gh13}
John Guckenheimer and Philip Holmes, \emph{Nonlinear oscillations, dynamical
  systems, and bifurcations of vector fields}, vol.~42, Springer Science \&
  Business Media, 2013.

\bibitem{Guna_threshold}
Jeremy Gunawardena, \emph{Multisite protein phosphorylation makes a good
  threshold but can be a poor switch}, P.\ Natl.\ Acad.\ Sci.\ USA \textbf{102}
  (2005), no.~41, 14617--14622.

\bibitem{hadac-osc}
Otto Hada{\v{c}}, Franti{\v{s}}ek Muzika, Vladislav Nevoral, Michal
  P{\v{r}}ibyl, and Igor Schreiber, \emph{Minimal oscillating subnetwork in the
  {Huang-Ferrell} model of the {MAPK} cascade}, PLOS ONE \textbf{12} (2017),
  no.~6, 1--25.

\bibitem{bistable}
Juliette Hell and Alan~D. Rendall, \emph{A proof of bistability for the dual
  futile cycle}, Nonlinear Anal.-Real \textbf{24} (2015), 175--189.

\bibitem{yeast-mapk-oscillations}
Zoe Hilioti, Walid Sabbagh, Saurabh Paliwal, Adriel Bergmann, Marcus~D
  Goncalves, Lee Bardwell, and Andre Levchenko, \emph{Oscillatory
  phosphorylation of yeast {Fus3} {MAP} kinase controls periodic gene
  expression and morphogenesis}, Curr.\ Biol. \textbf{18} (2008), no.~21,
  1700--1706.

\bibitem{oscillations-mapk-cancer}
Huizhong Hu, Alexey Goltsov, James~L Bown, Andrew~H Sims, Simon~P Langdon,
  David~J Harrison, and Dana Faratian, \emph{Feedforward and feedback
  regulation of the {MAPK} and {PI3K} oscillatory circuit in breast cancer},
  Cell.\ Signal. \textbf{25} (2013), no.~1, 26--32.

\bibitem{mathematica}
Wolfram~Research{,} Inc., \emph{Mathematica, {V}ersion 11.3}, Champaign, IL,
  2018.

\bibitem{imr17}
Brian Ingalls, Maya Mincheva, and Marc~R. Roussel, \emph{Parametric sensitivity
  analysis of oscillatory delay systems with an application to gene
  regulation}, B.\ Math.\ Biol. \textbf{79} (2017), no.~7, 1539--1563.

\bibitem{bi04}
Brian~P Ingalls, \emph{Autonomously oscillating biochemical systems: parametric
  sensitivity of extrema and period}, Systems biol. \textbf{1} (2004), no.~1,
  62--70.

\bibitem{translated}
Matthew~D. Johnston, \emph{Translated chemical reaction networks}, B.\ Math.\
  Biol. \textbf{76} (2014), no.~6, 1081--1116.

\bibitem{johnston-param}
Matthew~D. Johnston, Stefan M\"uller, and Casian Pantea, \emph{A
  deficiency-based approach to parametrizing positive equilibria of biochemical
  reaction systems}, Preprint, {\tt arXiv:1805.09295} (2018).

\bibitem{liu}
Wei~Min Liu, \emph{Criterion of {H}opf bifurcations without using eigenvalues},
  J.\ Math.\ Anal.\ Appl. \textbf{182} (1994), no.~1, 250--256. \MR{1265895}

\bibitem{lc12}
German Lozada-Cruz, \emph{The simple application of the implicit function
  theorem}, Boletin de la Asociati{\'o}n Matem{\'a}tica Venezolana
  \textbf{{XIX}} (2012), no.~1.

\bibitem{signs}
Stefan M\"uller, Elisenda Feliu, Georg Regensburger, Carsten Conradi, Anne
  Shiu, and Alicia Dickenstein, \emph{Sign conditions for injectivity of
  generalized polynomial maps with applications to chemical reaction networks
  and real algebraic geometry}, Found. Comput. Math. \textbf{16} (2016), no.~1,
  69--97.

\bibitem{ode}
Koji~L. Ode and Hiroki~R. Ueda, \emph{Design principles of
  phosphorylation-dependent timekeeping in eukaryotic circadian clocks}, Cold
  Spring Harbor Perspectives in Biology (2017).

\bibitem{PM}
Parag Patwardhan and W.~Todd Miller, \emph{Processive phosphorylation:
  Mechanism and biological importance}, Cell.\ Signal. \textbf{19} (2007),
  no.~11, 2218--2226.

\bibitem{messi}
Mercedes P{\'e}rez~Mill{\'a}n and Alicia Dickenstein, \emph{The structure of
  {MESSI} biological systems}, SIAM J.\ Appl.\ Dyn.\ Syst. \textbf{17} (2018),
  no.~2, 1650--1682.

\bibitem{TSS}
Mercedes {P\'erez Mill\'an}, Alicia Dickenstein, Anne Shiu, and Carsten
  Conradi, \emph{Chemical reaction systems with toric steady states}, B. Math.
  Biol. \textbf{74} (2012), no.~5, 1027--1065.

\bibitem{MPM_MAPK}
Mercedes P\'erez~Mill\'an and Adri\'an~G. Turjanski, \emph{{MAPK}'s networks
  and their capacity for multistationarity due to toric steady states}, Math.\
  Biosci. \textbf{262} (2015), 125--137.

\bibitem{Rao}
Shodhan Rao, \emph{Global stability of a class of futile cycles}, J.\ Math.\
  Biol. \textbf{74} (2017), 709--726.

\bibitem{Rao-2}
\bysame, \emph{Stability analysis of the {Michaelis--Menten} approximation of a
  mixed mechanism of a phosphorylation system}, Math.\ Biosci. \textbf{301}
  (2018), 159 --166.

\bibitem{long-term}
Boris~Y. Rubinstein, Henry~H. Mattingly, Alexander~M. Berezhkovskii, and
  Stanislav~Y. Shvartsman, \emph{Long-term dynamics of multisite
  phosphorylation}, Mol.\ Biol.\ Cell \textbf{27} (2016), no.~14, 2331--2340.

\bibitem{salazar}
Carlos Salazar and Thomas H\"ofer, \emph{Multisite protein phosphorylation --
  from molecular mechanisms to kinetic models}, FEBS Journal \textbf{276}
  (2009), no.~12, 3177--3198.

\bibitem{SK}
Thapanar Suwanmajo and J.~Krishnan, \emph{Mixed mechanisms of multi-site
  phosphorylation}, J. R. Soc. Interface \textbf{12} (2015), no.~107.

\bibitem{TG}
Matthew Thomson and Jeremy Gunawardena, \emph{The rational parameterisation
  theorem for multisite post-translational modification systems}, J. Theoret.
  Biol. \textbf{261} (2009), no.~4, 626--636.

\bibitem{ray}
Hwai-Ray Tung, \emph{Precluding oscillations in {M}ichaelis-{M}enten
  approximations of dual-site phosphorylation systems}, Preprint, {\tt
  arXiv:1712.03594} (2017).

\bibitem{beat}
David~M. Virshup and Daniel~B. Forger, \emph{Keeping the beat in the rising
  heat}, Cell \textbf{137} (2009), no.~4, 602--604.

\bibitem{yang-hopf}
Xiaojing Yang, \emph{Generalized form of {H}urwitz-{R}outh criterion and {H}opf
  bifurcation of higher order}, Appl.\ Math.\ Lett. \textbf{15} (2002), no.~5,
  615--621.

\end{thebibliography}

\end{document}